\documentclass[english, 11pt]{article}
\usepackage[T1]{fontenc}
\usepackage[utf8]{inputenc}
\usepackage[a4paper]{geometry}
\usepackage{color,babel,float}
\usepackage{amsmath,amsthm,amssymb,mathtools}
\usepackage[unicode=true,pdfusetitle,bookmarks=true,bookmarksnumbered=false,
bookmarksopen=false,breaklinks=true,pdfborder={0 0 1},backref=false,colorlinks=true]{hyperref}
\hypersetup{linkcolor=blue, urlcolor=blue, citecolor=blue}
\usepackage[x11names,usenames,dvipsnames,svgnames,table,rgb]{xcolor}
\usepackage[flushleft]{threeparttable}
\usepackage{tikz}
\usepackage{multirow}
\usepackage{epstopdf}
\usepackage{subfigure}
\usepackage{caption} 
\usepackage{dblfloatfix} 
\usepackage[title]{appendix}
\usepackage{authblk}
\usepackage{algorithmicx}
\usepackage{algpseudocode}
\usepackage{setspace}
\usepackage[backend=bibtex,style=numeric,sorting=nty,giveninits=true]{biblatex}
\theoremstyle{plain}
\ifx\thechapter\undefined
\newtheorem{thm}{\protect\theoremname}
\else
\newtheorem{thm}{\protect\theoremname}[chapter]
\fi
\theoremstyle{plain}
\newtheorem{prop}[thm]{\protect\propositionname}
\newtheorem{lemma}[thm]{\protect\lemmaname}

\floatstyle{ruled}
\newfloat{algorithm}{tbp}{loa}
\providecommand{\algorithmname}{Algorithm}
\floatname{algorithm}{\protect\algorithmname}

\usepackage{pdflscape}
\usepackage{lettrine}
\renewcommand\[{\begin{equation}}
\renewcommand\]{\end{equation}}
\newcommand{\real}{\mathbb{R}} 
\newcommand{\nn}{\mathbb{N}} 
\newcommand{\AAA}{\mathcal{A}}
\newcommand{\BBB}{\mathcal{B}}
\newcommand{\CCC}{\mathcal{C}}
\newcommand{\EEE}{\mathcal{E}}
\newcommand{\NNN}{\mathcal{N}}
\newcommand{\VVV}{\mathcal{V}}
\newcommand{\bone}{{\bf 1}}
\newcommand{\ba}{{\bf a}}
\newcommand{\be}{{\bf e}}
\newcommand{\bx}{{\bf x}}
\newcommand{\by}{{\bf y}}
\newcommand{\bo}{{\bf 0}}
\newcommand{\bs}{{\bf s}}
\newcommand{\bu}{{\bf u}}
\newcommand{\bv}{{\bf v}}
\newcommand{\bw}{{\bf w}}
\newcommand{\baa}{{\bf A}}
\newcommand{\bbb}{{\bf B}}
\newcommand{\bdd}{{\bf D}}

\newcommand{\bpp}{{\bf P}}
\newcommand{\bqq}{{\bf Q}}

\newcommand{\bww}{{\bf W}}
\newcommand{\bxi}{{\boldsymbol{\xi}}}

\newcommand{\norm}[1]{\left\Vert {#1} \right\Vert} 
\newcommand{\act}[1]{\left\langle {#1} \right\rangle} 
\newcommand{\Seq}[2]{\{{#1}^{{#2}}\}_{{#2} \in \mathbb{N}}}
\DeclareMathOperator*{\argmin}{argmin}
\DeclareMathOperator*{\Ker}{ker}
\DeclareMathOperator*{\Span}{span}

\theoremstyle{plain}
\newtheorem{assumption}{Assumption}

\makeatother

\providecommand{\propositionname}{Proposition}
\providecommand{\theoremname}{Theorem}
\providecommand{\lemmaname}{Lemma}

\begin{document}

\title{Alternating Minimization Based First-Order Method for the Wireless Sensor Network Localization Problem\thanks{This research was partially supported by Israel Science Foundation Grant 1460/19.}}
\author{Eyal Gur}
\author{Shoham Sabach}
\author{Shimrit Shtern}

\affil{\small Faculty of Industrial Engineering and Management, Technion - Israel Institute of Technology, Technion city, Haifa 3200003, Israel\\
	\href{mailto:eyal.gur@campus.technion.ac.il}{eyal.gur@campus.technion.ac.il}, \href{mailto:ssabach@ie.technion.ac.il}{ssabach@ie.technion.ac.il}, \href{mailto:shimrits@technion.ac.il}{shimrits@technion.ac.il}}
\date{}

\maketitle

\doublespacing
\begin{abstract}
	We propose an algorithm for the Wireless Sensor Network localization problem, which is based on the well-known algorithmic framework of Alternating Minimization. We start with a non-smooth and non-convex minimization, and transform it into an equivalent smooth and non-convex problem, which stands at the heart of our study. This paves the way to a new method which is globally convergent: not only does the sequence of objective function values converge, but the sequence of the location estimates also converges to a unique location that is a critical point of the corresponding (original) objective function. The proposed algorithm has a range of fully distributed to fully centralized implementations, which all have the property of global convergence. The algorithm is tested over several network configurations, and it is shown to produce more accurate solutions within a shorter time relative to existing methods.
\end{abstract}


\section{Introduction} \label{Sec:Intorduction}
	S{ensor} networks consist of several wireless sensors located in a given area, for purposes such as environment monitoring, battlefields surveillance, etc. (see \cite{ammari2014art} for more examples and details). In our setting, each sensor is composed of a low-powered radio transceiver which monitors its immediate surroundings (e.g., temperature, sound, etc.), and a processor that collects and manipulates the data. Therefore, the location of each sensor has a significant role. Since the number of sensors in a network can be large (even thousands of sensors), it is not cost effective to equip each one of them with a GPS device, nor to deploy the sensors in a known logged location.

	A network in this context is described as a group of $K$ sensors, each denoted by an index in the set $\left\{ 1 , 2 , \ldots , K \right\}$. The communication between two sensors $i$ and $j$ is made available only if the distance between the two is at most $r \geq 0$, which is a given radio range for communication. In this case, we say that the sensors $i$ and $j$ are neighbors. The Wireless Sensor Network (WSN) localization problem aims at finding the location of all sensors in the network based on a few anchors, which are sensors with a known location (e.g., via GPS devices), and noisy distance measurements from each sensor to its neighbors. Mathematically speaking, given $K$ sensors and anchors in $\real^{n}$, where $m$ of which ($m < K$) are anchors, we wish to estimate the location of all $N = K - m$ sensors in $\real^{n}$. 

	We denote, for simplicity, the set of the $N$ sensors to be located by $\VVV := \left\{ 1 , 2 , \ldots , N \right\}$, and the set of the $m$ anchors by $\AAA := \left\{ N + 1 , N + 2 , \ldots , K \right\}$. For an anchor $j \in \AAA$ we denote its given location by $\ba_{j} \in \real^{n}$. Additionally, we denote by $\EEE$ the set of all pairs of neighboring sensors, i.e., $\left(i , j\right) \in \EEE$ if $i < j$ and sensors $i$ and $j$ are neighbors, which means that the distance between them is at most $r$. For each pair $\left(i , j\right) \in \EEE$, the (noisy) measurement of their distance is denoted by $d_{ij}$ (following \cite{CXG2015}, we assume w.l.o.g. that $d_{ij} = d_{ji}$)\footnote{Note that, in this paper, we use the standard assumption that all distance measurements for sensor pairs which are within the communication radius are available.}. We use the notation $M = \lvert \EEE \rvert$. Note that since the locations of the anchors are known exactly, we can ignore all edges between anchor nodes, which means that if $\left(i , j\right)\in\EEE$ then $i\in \VVV$. Moreover, we denote by $\bx_{i} \in \real^{n}$ the true location of sensor $i \in \left\{ 1 , 2 , \ldots , N \right\}$, and we denote by $\bx \in \real^{nN}$ the vector obtained by concatenating the vectors $\bx_{i}$ for all $i \in \left\{ 1 , 2 ,\ldots , N \right\}$ into a single column vector. Using these notations, in this work we adopt the following non-smooth and non-convex formulation of the WSN localization problem:
	\begin{equation}
		\underset{\bx \in \real^{nN}}{\min} \left\{ \sum_{\left(i , j\right) \in \EEE_{1}} \left(\norm{\bx_{i} - \bx_{j}} - d_{ij}\right)^{2} +\sum_{\left(i , j\right) \in \EEE_{2}} \left(\norm{\bx_{i} - \ba_{j}} - d_{ij}\right)^{2}\right\}, \label{Problem}
\end{equation}	
	where $\EEE_{1}$ is the subset of $\EEE$ with all pairs $\left(i , j\right)$ for which both sensors $i$ and $j$ are non-anchors, while $\EEE_{2}$ is the subset of all pairs $\left(i , j\right)$ in $\EEE$ for which sensor $i$ is a non-anchor and sensor $j$ is an anchor. 
	
	This formulation enjoys a statistical interpretation as of finding the maximum-likelihood estimator of the locations, given that the measurement noises are independently normally (Gaussian) distributed (see, for example, \cite{press2007numerical}).
	
	It should be noted that a closely related problem is the Single Source Localization (SSL), which can be seen as a particular case of the WSN problem for $N = 1$. However, the network variant possess several challenges which do not exist in the SSL problem both from an algorithmic and a computational standpoints as we highlight below (for instance, distributed and/or parallel implementations).
	
\subsection{Existing Methods for Convex Relaxations} \label{SSec:Relaxations}
	One common approach for tackling the non-convex Problem \eqref{Problem}, is by solving certain convex relaxations of the problem\footnote{The literature on convex relaxations is extensive, however, since this is not the focus of our paper we do not give a thorough review of this approach, but rather highlight few ideas and their limitations.}. For example, in \cite{CXG2015}, the authors show that the first summed terms in Problem \eqref{Problem} are just the squared distance in $\real^{n}$, of the point ${\bx_{i} - \bx_{j}}$ to the sphere with radius $d_{ij}$ centered at the origin (similarly for the other sum in the objective). Thus, the convexification presented in \cite{CXG2015} is derived by taking the squared distance to the ball of radius $d_{ij}$, instead to the sphere. This yields a smooth and convex problem with separable ball constraints that can be solved by classical techniques of convex optimization, especially the well-known Accelerated Projected Gradient method of Nesterov \cite{N1983}.

	Another popular relaxation, which stands at the basis of many works in this area, is based on the classical ``lifting'' technique, in which the non-convexity in the objective is replaced by quadratic equality constraints. These constraints are then relaxed to inequality constraints, and the problem is reformulated as a Semi-Definite Programming (SDP) problem (see, for instance, \cite{BLWY2006,simonetto2014distributed,lui2009semi} and references therein). These SDP problems can then be solved using off-the-shelf external solvers that utilize interior-point methods, and are known to be computationally expensive. Thus, these techniques may suffer from an increased computational complexity, which is less desirable when dealing with large-scale problems, where thousands of sensors are deployed (it should be remembered that the "lifting" process significantly increases the dimension of the problem to be solved). Thus, in many cases a further relaxation of the SDP constraints, called edge-based SDP \cite{WZYS2008}, is used in order to obtain a distributed algorithm, which is less computationally demanding (see, for instance, \cite{simonetto2014distributed} and \cite{lui2009semi}). However, this still requires the implementation of an interior-point method on each sensor, which could be a major limitation due to the small computational power of individual sensors.

	While convex relaxations guarantee convergence to a global minimum point, this point is not necessarily a global minimizer, or even a critical point, of the original (non-relaxed) formulation. See \cite{PE2018} for a discussion on this phenomenon. In order to demonstrate this phenomenon, we have also conducted a numerical comparison between our method, which is applied to the original non-convex and non-smooth formulation of the problem, to the relaxation approach suggested in \cite{CXG2015} (see Section \ref{Sec:Numerics}). Indeed, due to the superior estimation quality of our method compared to the relaxation, as well as similar evidences in previous papers, we do not compare our method to other methods which aim at solving any kind of convex relaxation of the problem.

	An alternative formulation used in the literature, that also allows for inaccurate location estimates of the anchor sensors (see, for example, \cite{BLTYW2006} and references therein), is given by the following smooth and non-convex problem:
	\begin{equation} 
		\underset{\bx \in \real^{nN}}{\min} \left\{ \sum_{\left(i , j\right) \in \EEE_{1}}\left(\norm{\bx_{i} - \bx_{j}}^{2} - d_{ij}^{2}\right)^{2} + \sum_{\left(i , j\right) \in \EEE_{2}} \left(\norm{\bx_{i} - \ba_{j}}^{2} - d_{ij}^{2}\right)^{2}\right\}. \label{SquaredFormulation}
\end{equation}	
	While this formulation is smooth, it lacks the statistical interpretation of Problem \eqref{Problem}.	Additionally, it is suggested in \cite{BSL2008}, that at least for the Single Source Localization problem, solving formulation \eqref{Problem} produces more accurate estimations when compared with formulation \eqref{SquaredFormulation}. Moreover, it does not change the fact that Problem \eqref{SquaredFormulation} remains non-convex. Therefore, also in this case convex relaxations, such as SDP relaxations, are often used. However, solving these relaxations come with drawbacks, which have already been discussed above. {Therefore, even though many methods were developed to tackle this formulation, we focus our study on the formulation \eqref{Problem} and provide a comparison between methods that tackles this model.}

	It is worth mentioning that the literature on the WSN localization problem also includes heuristic approaches (see, for example, \cite{NM2009}, \cite{GP2016} and \cite{agarwal2012sensor}). However, these approaches have either weak or no theoretical guarantees on both the convergence of the sequence and quality of the resulting solution. 

\subsection{Related First-Order Methods} \label{SSec:First}
	In this paper, we aim at solving the original unconstrained non-smooth and non-convex Problem \eqref{Problem} directly, using first-order methods, namely methods that exploit information on values and (sub)gradients of the involved functions. We now focus on three relevant works \cite{PE2018}, \cite{SHCJ2010} and \cite{SJJ2014}, which also propose, among other things, first order methods for solving Problem \eqref{Problem}. In these works, the first-order method to solve Problem \eqref{Problem} is used to improve the localization accuracy of solutions that were obtained from solving some convex relaxation of Problem \eqref{Problem}, see Section \ref{SSec:Relaxations} for a discussion about convex relaxations. The description of this two-stage strategy is postponed to Section \ref{Sec:Numerics}. At this moment, we focus the discussion on the first-order algorithms as standalone optimization algorithms for solving Problem \eqref{Problem}.
		
	In \cite{PE2018}, a linear constrained equivalent reformulation of Problem \eqref{Problem} is studied and a fully distributed ADMM method was suggested to solve the non-convex reformulated problem. The theoretical result of this work, which is based on \cite{ABMS2007}, says that the generated sequence convergence to a local minima\footnote{We have found this result misleading. In \cite{ABMS2007} the authors provides conditions, in the non-convex setting, for a sequence generated by ADMM to have limit points, and for these limit points to be KKT points, however no guarantees for local minima are established there.}. While this algorithm performance seems to be promising, it posses some practical  challenges: (i) tuning of several parameters, (ii) a sub-algorithm to solve the non-convex problems with respect to the non-anchor sensors and (iii) the algorithm may generate iterations which are not well-defined. We will discuss these limitations further in Section \ref{Sec:Numerics}.

	In \cite{SHCJ2010}, a reformulation of Problem \eqref{Problem} as a constrained minimization is considered\footnote{It should be noted that the authors of \cite{SHCJ2010} study a more general form of Problem \eqref{Problem}, which allows for the anchors to have noisy location estimates.}, similar to the reformulation that we suggest below. This reformulation requires the definition of new variables, and immediately suggests an Alternating Minimization based algorithm that they call Non-Convex Sequential Greedy method. The obtained algorithm, as the algorithm of \cite{PE2018}, is not always well-defined (as we discuss below). This is a fully distributed algorithm that is shown to generate a converging sequence of function values. The authors also show that any limit point of the generated sequence is a KKT point of the reformulated problem (but no relations to the original problem are provided)
	
	In \cite{SJJ2014}, the authors reformulate the problem as in \cite{CXG2015}, but instead of employing a convex relaxation, they aim at solving the reformulated non-convex problem using the classical Projected Gradient algorithm. However, the paper does not provide any theoretical guarantees about their algorithm.
	
\subsection{Our Contribution}
	As we wish to solve the original non-smooth formulation given in Problem \eqref{Problem} using first-order methods, we first address the inherent non-smoothness of the problem. We adopt a simple variational representation of the Euclidean norm, which enables us to transform Problem \eqref{Problem} into an equivalent constrained smooth and non-convex minimization problem. This formulation, which is closely related to the formulation used in \cite{SHCJ2010}, will serve as a starting point for our study.

	Our first contribution is an algorithmic framework for solving the original non-convex Problem \eqref{Problem}, that includes the whole spectrum ranging from a fully centralized to a fully distributed algorithm. The concept of Alternating Minimization (AM), which stands at the heart of our framework, exploits the structure of our reformulation, to produce a well-defined and parameter-free algorithm. Moreover, our second contribution is proving theoretical guarantees of this algorithm, which are stronger in comparison to those obtained in the two mentioned works \cite{PE2018} and \cite{SHCJ2010}. More precisely, we prove global convergence of the generated sequence of network location estimates (in the sense that the whole sequence converges), that is, the location estimates of each sensor converges to a unique limit point\footnote{Since we are dealing with a non-convex problem, this point obviously depends on the starting point of the algorithm.}. In addition, we show that this unique limit point of the network location estimates is a critical point of the function under minimization in Problem \eqref{Problem}, {where a point is a critical point if its sub-differential contains the zero vector}. In order to do so, we use a recent proof technique, which was proposed in \cite{AB2009} and \cite{ABS2013}, and later was unified and simplified in \cite{BST2014} (see also \cite{BSTV2018} for a recent concise description of this technique). To the best of our knowledge, this is the first work that proves a global convergence result for a first-order algorithm that solves the original non-convex WSN problem.

	As mentioned above, our algorithm implementation can range from fully distributed, i.e., each sensor updates its own location estimate through information received over communication, to fully centralized, where the location estimate of the entire network is calculated on one central processor.
	
	Additionally, we show that the fully distributed version of our algorithm, to be developed below, can be seen as a modification of the algorithm presented in \cite{SHCJ2010}, which is always well-defined, and enjoys our stronger convergence results.

	The paper is organized as follows: the smooth and non-convex equivalent formulation of Problem \eqref{Problem} is introduced in Section \ref{Sec:Formulation}. In Section \ref{Sec:Algorithms}, we introduce the fully centralized and fully distributed algorithms for tackling Problem \eqref{Problem}. We present the Unifying AM Algorithm that captures both centralized and distributed versions, in Section \ref{Sec:Unifying}, and discuss the capability of this unified approach to also capture partially-distributed and partially-parallelized algorithmic variations. We analyze the convergence of the Unifying AM Algorithm in Section \ref{Sec:Analysis}. In Section \ref{Sec:Numerics}, we conduct numerical experiments to evaluate the performance of our algorithms and compare them with existing methods. 

\section{Problem Formulation} \label{Sec:Formulation}
	In the introduction above, we have already mentioned that since Problem \eqref{Problem} is non-smooth, our first step is to derive an equivalent smooth reformulation of the problem, where by equivalent we mean that a critical point of the reformulated problem is also a critical point of Problem \eqref{Problem}, and vice versa.

	Notice that the first sum in the objective function of Problem \eqref{Problem} can be written explicitly as				
	\begin{equation} \label{NonsmoothTerm}
		\sum_{\left(i , j\right) \in \EEE_{1}}\left(\norm{\bx_{i} - \bx_{j}}^{2} - 2d_{ij}\norm{ \bx_{i} - \bx_{j}} + d_{ij}^{2}\right).
	\end{equation}
	The non-smoothness of \eqref{NonsmoothTerm} comes from the terms $\norm{\bx_{i} - \bx_{j}}$ (the true distance between sensors $i$ and $j$) for all $\left(i , j\right) \in \EEE_{1}$. To eliminate this, we first denote by $\BBB \equiv B_{1}\left(\bo_{n}\right)$ the unit ball in $\real^{n}$ centered at the origin, and by $\BBB^M$ the Cartesian product of $M$ such balls, where $M$ is the number of edges in the network. Motivated by the recent work \cite{LSTZ2017} that focuses on the Single Source Localization problem, we use the following fact on the WSN problem, which easily follows from the Cauchy-Schwartz inequality:
	\begin{equation*}
		\norm{\bx_{i} - \bx_{j}} = \max\left\{ \bu_{ij}^{T}\left(\bx_{i} - \bx_{j}\right) : \, \bu_{ij} \in \BBB \right\},
	\end{equation*}
	where $\bu_{ij} \in \real^{n}$ is an auxiliary variable defined for each pair $\left(i , j\right) \in\EEE_{1}$, that achieves the optimal value of $\norm{\bx_{i} - \bx_{j}}$ at $\left(\bx_{i} - \bx_{j}\right)/\norm{\bx_{i} - \bx_{j}}$ when $\bx_{i} - \bx_{j} \neq \bo_{n}$ (otherwise any $\bu_{ij} \in \BBB$ would be an optimal solution, but in this paper we will always take $\bu_{ij} = \bo_{n}$).
	
	Similar manipulations can be done on the second sum in the objective function of Problem \eqref{Problem}. Therefore, by omitting the constant terms $d_{ij}^{2}$, $\left(i , j\right) \in \EEE$, we arrive at the following smooth and constrained equivalent reformulation of Problem \eqref{Problem}:
	\begin{equation}
		\underset{\begin{smallmatrix}\bx \in \real^{nN} \bu \in \BBB^{M}\end{smallmatrix}}{\min} \left\{ \sum_{\left(i , j\right) \in \EEE_{1}} \left(\norm{\bx_{i} - \bx_{j}}^{2} - 2d_{ij}\bu_{ij}^{T}\left(\bx_{i} - \bx_{j}\right)\right)+ \sum_{\left(i , j\right) \in \EEE_{2}}\left(\norm{\bx_{i} - \ba_{j}}^{2} - 2d_{ij}\bu_{ij}^{T}\left(\bx_{i} - \ba_{j}\right)
		\right)\right\}, \label{ObjG}
	\end{equation}
	where $\bu$ is the vector obtained by concatenating the vectors $\bu_{ij}$ for all $\left(i , j\right) \in \EEE$ into a single column vector.
	
	For the sake of simplifying the developments to come, we wish to rewrite Problem \eqref{ObjG} using matrix notations. To this end, one could notice that the network of the original problem can be viewed as an undirected graph, where the vertices are the sensors, and $\EEE$ is the set of edges. For simplicity, we fix some ordering of the set $\EEE$ such that all edges in $\EEE_{1}$ proceed the edges in $\EEE_{2}$. We now define several essential matrices for the formulation of the WSN localization problem:
	\begin{itemize}
		\item $\tilde{\bqq} \in \real^{|\EEE_{1}| \times N}$ is the arc-node incidence matrix of the subgraph containing only the vertices in $\VVV$ and edges in $\EEE_{1}$, i.e., if $\left(i , j\right) \in \EEE_{1}$ is the $l$-th edge in the set $\EEE_{1}$, then $\tilde{\bqq}_{li} = 1$, $\tilde{\bqq}_{lj} = -1$ and all other entries in the $l$-th row are equal to $0$.
		\item $\tilde{\baa} \in \real^{|\EEE_{2}| \times N}$ is the indicator matrix of the set of edges $\EEE_{2}$, i.e., $\tilde{\baa}_{li} = 1$ if the $l$-th edge of $\EEE_{2}$ connects the sensor $i \in \VVV$ to an anchor in $\AAA$ and all other entries in the $l$-th row are equal to $0$.
		\item  $\tilde{\bbb} \in \real^{|\EEE_{2}| \times m}$ is such that $\tilde{\bbb}_{lj} = -1$ if the $l$-th edge of $\EEE_{2}$ connects a certain sensor in $\VVV$ to the anchor $j + N \in \AAA$. It should be noted that the matrices $\tilde{\baa}$ and $\tilde{\bbb}$ form together an arc-node incidence matrix of the subgraph containing the vertices in $\VVV \cup \AAA$ and edges in $\EEE_{2}$.
		\item $\tilde{\bdd} \in \real^{M \times M}$ is the diagonal matrix with $d_{ij}$, $\left(i,j\right) \in \EEE$, as the entries on the diagonal.
		\item $\tilde{\bpp} \equiv \tilde{\bqq}^{T}\tilde{\bqq} + \tilde{\baa}^{T}\tilde{\baa} \in \real^{N\times N}$ (note that the matrix $\tilde{\bqq}^{T}\tilde{\bqq}$ is the so-called Laplacian matrix of the corresponding subgraph).
	\end{itemize}
	Since each sensor is in $\real^{n}$ it will be convenient to define the Kronecker matrix product of $\tilde{\baa}$, $\tilde{\bbb}$, $\tilde{\bdd}$, $\tilde{\bpp}$ and $\tilde{\bqq}$ with the identity matrix $\mathbf{I}_{n}$, which are denoted by $\baa$, $\bbb$, $\bdd$, $\bpp$ and $\bqq$, respectively. 

	Finally, denoting by $\ba \in \real^{nm}$, the vector derived by concatenating the vectors $\ba_{j}$, $j \in \AAA$, into a single column vector, we see that Problem \eqref{ObjG} can be written equivalently as:
	\begin{equation} \label{ReformulationMatrix}
		\underset{\begin{smallmatrix}\bx \in \real^{nN} \\ \bu \in \BBB^{M}\end{smallmatrix}}{\min} \left\{ \bx^{T}\bpp\bx - 2\left(\bww\bu + \baa^{T}\bbb\ba\right)^{T}\bx + 2\bs^{T}\bu \right\},
	\end{equation}
	where we define the matrix $\bww = \left[\bqq^{T} , \baa^{T}\right]\bdd$ and the vector $\bs = \left(\left[\bo_{nm \times n|\EEE_{1}|} , \bbb^{T}\right]\bdd\right)^{T}\ba$.

	We denote the objective function of Problem \eqref{ReformulationMatrix} by $G\left(\bx , \bu\right)$. We will also use the following function which takes the constraints on the block $\bu$ via its indicator formulation
	\begin{equation} \label{ObjF}
		F\left(\bx , \bu\right) := G\left(\bx , \bu\right)  + \sum_{\left(i , j\right) \in \EEE} \delta_{\BBB}\left(\bu_{ij}\right),
	\end{equation}
	where $\delta_{\BBB}\left(\cdot\right)$ denotes the indicator function of the set $\BBB$ (which is defined to be zero in $\BBB$ and $+\infty$ outside). This results in the following unconstrained optimization problem
	\begin{equation} \label{ReformulationCompact}
		\min \left\{ F\left(\bx , \bu\right) : \, \bx \in \real^{nN}, \, \bu \in \real^{nM} \right\}.
	\end{equation}
	
	This formulation will be the starting point of our study. We are interested in designing a simple and parameter-free algorithm for solving Problem \eqref{Problem} via its equivalent reformulation \eqref{ReformulationCompact}, which globally converges to a critical point of the original problem \eqref{Problem}. To this end, we will show that our algorithm finds a critical point of $F\left(\cdot , \cdot\right)$, that is, a pair $\left(\bx , \bu\right)$ which satisfies
	\begin{equation*}
		\left(\bo_{nN} , \bo_{nM}\right) \in \partial F\left(\bx , \bu\right) = \left\{ \nabla_{\bx} F\left(\bx , \bu\right) \right\} \times \left\{ \partial_{\bu} F\left(\bx , \bu\right) \right\},
	\end{equation*}
	where $\partial \varphi$ denotes the classical subdifferential of $\varphi$, that can be used here since the function $\bu \rightarrow F\left(\bx , \bu\right)$ is convex for any fixed $\bx \in \real^{nN}$.
	
	In order to simplify the inclusion above we will first need the following additional notation. We denote by $\NNN(i)$ the set of sensor $i$'s neighbors, i.e., $j \in \NNN(i)$ means that $\left(i , j\right) \in \EEE$ or\footnote{Since $\EEE$ includes only pairs $\left(i , j\right)$ for which $i < j$, we need to consider also the pairs where $j < i$.} $\left(j , i\right) \in \EEE$. We also denote $M_{i} = \left| \NNN(i) \right|$ and thus $M = \frac{1}{2} \sum_{i = 1}^{N} M_{i}$. 

	Therefore, the inclusion above easily translates into the following conditions:
	\begin{equation} \label{OptX}
		\bo_{n} = \hspace{-0.1in} \sum_{j \in \NNN(i) \cap \VVV} \hspace{-0.1in}\left(\bx_{i} - \bx_{j} - d_{ij}\bu_{ij}\right) + \hspace{-0.15in}\sum_{j \in \NNN(i) \cap \AAA} \hspace{-0.1in}\left(\bx_{i} - \ba_{j} - d_{ij}\bu_{ij}\right), 
	\end{equation}
	for all $i \in \left\{ 1 , 2 , \ldots , N \right\}$ and
	\begin{align}
		\bo_{n} & \in -2d_{ij}\left(\bx_{i} - \bx_{j}\right) + \partial \delta_{\BBB}\left(\bu_{ij}\right), \quad \,\, \forall \,\,  \left(i , j\right) \in \EEE_{1}, \label{OptU1} \\
		\bo_{n} & \in -2d_{ij}\left(\bx_{i} - \ba_{j}\right) + \partial \delta_{\BBB}\left(\bu_{ij}\right), \quad \,\, \forall \,\,  \left(i , j\right) \in \EEE_{2}, \label{OptU2}
	\end{align}
	where we define $\bu_{ji} = -\bu_{ij}$ for all $\left(i , j\right) \in\EEE_{1}$. In the Appendix below we prove that a vector $\bx$, which satisfies these inequalities, must be a critical point of the original problem \eqref{Problem} as recorded in the following result.
	\begin{prop} \label{P:CriticalCoincide}
		Let $\left(\bx^{\ast} , \bu^{\ast}\right)$ be a critical point of Problem \eqref{ReformulationCompact}, then $\bx^{\ast}$ is a critical point of the original problem \eqref{Problem}.
	\end{prop}	
	In the forthcoming section we develop a simple solution scheme for solving Problem \eqref{ReformulationCompact}, which globally converges to a pair that satisfies the premises of Proposition \eqref{P:CriticalCoincide} and therefore finds a critical point of the original WSN problem. To the best of our knowledge, this is the first algorithm that is guaranteed to converge to a critical point of the original non-smooth and non-convex problem (in the papers mentioned in Section \ref{SSec:First} only \textit{subsequences convergence} to critical/KKT points of the \textit{reformulated problem} is proven).
	 	 
\section{Centralized and Distributed Algorithms} \label{Sec:Algorithms}
	The two blocks structure of Problem \eqref{ReformulationCompact} (in addition to the fact that there is no coupling constraint between these blocks) immediately suggests the application of optimization methods that employ the concept of Alternating Minimization (AM), which is a very useful technique to tackle complex convex and non-convex problems. The main reason for that is the ability to exploit the following nice feature of Problem \eqref{ReformulationCompact}: each sub-problem with respect to one block of variables (while the other block remains fixed) can be easily solved.

	In the context of Problem \eqref{ReformulationCompact}, a basic AM based algorithm will have the following form
\smallskip

	\begin{center}
		\fcolorbox{black}{white}{\parbox{16cm}{{\bf A Centralized Alternating Minimization for WSN} \\
			{\bf Initialization.} $\bu^{0} \in \real^{nM}$ \\
	  		{\bf General Step.} For $k \in \nn$,
				\begin{itemize}
					\item[1.] {\bf Sensors update.} 
						\begin{align} \label{AM:Stepx}
							\bx^{k + 1} = \argmin_{\bx \in \real^{nN}} F\left(\bx , \bu^{k}\right) = \bpp^{-1}\left(\bww\bu + \baa^{T}\bbb\ba\right)^{T}.
						\end{align}
						\vspace{-0.5in}
					\item[2.] {\bf Auxiliary update.} For all $\left(i , j\right) \in \EEE$ we denote
						\begin{equation*}
							\bv_{ij}^{k + 1} = 
							\begin{cases}
								\bx_{i}^{k + 1} - \bx_{j}^{k + 1}, & \left(i , j\right) \in \EEE_{1}, \\
								\bx_{i}^{k + 1} - \ba_{j}, & \left(i , j\right) \in \EEE_{2}.
							\end{cases}
						\end{equation*}					
						Then
						\begin{equation} \label{AM:Stepu}
						\begin{aligned}
							\bu_{ij}^{k + 1} & = 
							\begin{cases}
								\frac{\bv_{ij}^{k + 1}}{\norm{\bv_{ij}^{k + 1}}}, & \bv_{ij}^{k + 1} \neq \bo_{n}, \\
								\bo_{n}, & \text{otherwise}.
							\end{cases}\\							
							\bu_{ji}^{k + 1}&=-\bu_{ij}^{k + 1}
							\end{aligned}
						\end{equation}					
				\end{itemize}\vspace{-0.1in}}}
	\end{center}
\smallskip

	The sub-problem in \eqref{AM:Stepx} is solved explicitly by solving the linear equation $\nabla_{\bx} F\left(\bx^{k + 1} , \bu^{k}\right) = \bo$ (in this respect see also Proposition \ref{P:BasicPro} below). Therefore, the obtained solution (see exact formula above) means that this algorithm is centralized in the sense that the location update rule \eqref{AM:Stepx} requires to perform the computation on a single processing unit. The problem of minimizing the function $F\left(\bx^{k + 1} , \bu\right)$ with respect to $\bu$ consists of separable minimization problems (see \eqref{ObjG}), each of minimizing a linear function over the unit ball, which results with the formula given in \eqref{AM:Stepu}. Before proceeding, we note that this algorithm generates a well-defined sequence (this is a property that the algorithms mentioned in the Introduction do not necessarily possess, more on that in Section \ref{Sec:Numerics}).
	
	A closer look on Problem \eqref{ReformulationCompact} reveals that the AM approach can be used in a more refined way to produce a fully distributed algorithm. More precisely, we would like to allow each sensor to update its own location estimate, and to this end, we will look at each sub-block $\bx_{i}$, $i \in \left\{ 1 , 2 ,\ldots , N \right\}$, as a separated block of variables and employ the AM approach on the $N + 1$ blocks: $\bx_{1} , \bx_{2} , \ldots , \bx_{N}$ and $\bu$, as explicitly recorded now.
\smallskip

	\begin{center}
		\fcolorbox{black}{white}{\parbox{16cm}{{\bf A Distributed Alternating Minimization for WSN} \\
			{\bf Initialization.} $\bx^{0} \in \real^{nN}$ and $\bu^{0} \in \real^{nM}$ \\
 	 		{\bf General Step.} For $k \in \nn$,
				\begin{itemize}
					\item[1.] {\bf Sensors update.} For all $i \in \left\{ 1 , 2 , \ldots , N \right\}$
						\begin{equation} \label{eq:update}
							\hspace{-0.25in} \bx_{i}^{k + 1} = \frac{1}{M_{i}}\left(\sum_{\underset{j < i}{j \in \NNN(i)} \cap \VVV} \hspace{-0.1in} \bx_{j}^{k + 1} + \hspace{-0.1in} \sum_{\underset{j > i}{j \in \NNN(i)} \cap \VVV} \hspace{-0.1in} \bx_{j}^{k} + \hspace{-0.1in} \sum_{j\in \NNN(i)} \hspace{-0.05in} d_{ij}\bu_{ij}^{k} + \hspace{-0.15in} \sum_{l \in \NNN(i) \cap \AAA} \hspace{-0.1in} \ba_{l} \right).
						\end{equation}
					\vspace{-0.3in}
					\item[2.] {\bf Auxiliary update.} The same as in \eqref{AM:Stepu}.
				\end{itemize}}}
	\end{center}
\smallskip

	The update rule of the sensors as given in \eqref{eq:update} easily follows from writing the optimality condition of minimizing the function $\bx_{i} \rightarrow F\left(\bx_{1}^{k + 1} , \ldots , \bx_{i - 1}^{k + 1} , \bx_{i} , \bx_{i + 1}^{k} , \ldots , \bx_{N}^{k}\right)$ (see \eqref{OptX} in this respect).
	
	Note that while this algorithm is fully distributed, it is serial in updating the blocks $\bx_{1} , \bx_{2} , \ldots , \bx_{N}$. However, parallelization can be achieved in the updating of the auxiliary variables $\bu_{ij}$, $\left(i , j\right) \in \EEE$. In the upcoming section we discuss how one can further parallelize this algorithm when updating the blocks $\bx_{1} , \bx_{2} , \ldots , \bx_{N}$.

\section{A Unifying AM Algorithm for WSN} \label{Sec:Unifying}
	The centralized and distributed algorithms presented above are both particular instances of our Unifying AM Algorithm, which is developed next. The idea is to divide the set of sensors into $q \in \left\{ 1 , 2 , \ldots , N \right\}$ disjoint clusters $\CCC_{1} , \CCC_{2} , \ldots , \CCC_{q}$ such that $\CCC_{1} \cup \CCC_{2} \cup \cdots \cup \CCC_{q} = \left\{ 1 , 2 , \ldots , N \right\}$, thus forming a partition of the set of sensors $\VVV$. It is easy to see that the fully centralized version can be recovered when the number of clusters is $q = 1$, while the fully distributed version is recovered when $q = N$.
	
	In order to derive an AM based algorithm that exploits the division of the block $\bx$ into clusters we will need to denote sub-vectors with respect to this division. For each cluster $\CCC_{i}$, $1 \leq i \leq q$, we denoted by ${\bar \bx}_{i}$ the vector which is constructed by concatenating the vectors $\bx_{j}$ for all $j \in \CCC_{i}$. Therefore, we can now apply the AM technique to Problem \eqref{ReformulationCompact}, with respect to the $q + 1$ blocks: ${\bar \bx}_{1} , {\bar \bx}_{2} , \ldots , {\bar \bx}_{q}$ and $\bu$. More precisely, the algorithm is recorded now in an abstract form (explicit formulas can be easily derived by writing the corresponding optimality conditions).
\smallskip

	\begin{center}
		\fcolorbox{black}{white}{\parbox{16cm}{{\bf A Unifying Alternating Minimization for WSN} \\
			{\bf Initialization.} $\bx^{0} \in \real^{nN}$ and $\bu^{0} \in \real^{nM}$ \\
	  		{\bf General Step.} For $k \in \nn$,
				\begin{itemize}
					\item[1.] {\bf Sensors update.} For all $i \in \left\{ 1 , 2 , \ldots , N \right\}$
						\vspace{-0.1in}	
						\begin{align*}
							{\bar \bx}_{1}^{k + 1} & = \argmin_{{\bar \bx}_{1}} F\left({\bar \bx}_{1} , {\bar \bx}_{2}^{k} , \ldots , {\bar \bx}_{q}^{k} , \bu^{k}\right), \\
							{\bar \bx}_{2}^{k + 1} & = \argmin_{{\bar \bx}_{2}} F\left({\bar \bx}_{1}^{k + 1} , {\bar \bx}_{2} , \ldots , {\bar \bx}_{q}^{k} , \bu^{k}\right), \\ 
					\vdots & \\
							{\bar \bx}_{q}^{k + 1} & = \argmin_{{\bar \bx}_{q}} F\left({\bar \bx}_{1}^{k + 1} , {\bar \bx}_{2}^{k + 1} , \ldots , {\bar \bx}_{q} , \bu^{k}\right).
						\end{align*}
					\vspace{-0.5in}	
					\item[2.] {\bf Auxiliary update.} The same as in \eqref{AM:Stepu}.
				\end{itemize}}}
	\end{center}
\smallskip

	Before analyzing the Unifying AM Algorithm presented above, we would like to provide another advantage of the idea of dividing the sensors into disjoint clusters. 
	
	There are many ways in which one may divide the sensors into clusters, each with its own benefits. The most naive approach would be \textit{geographical clustering}, in which several neighboring sensors are collected into one cluster. This division makes the algorithm more centralized, since the updates of each cluster take into account the interrelations between the sensors that comprise it. Alternatively, one may use \textit{colored clustering}, i.e., clusters that consist of non-neighboring sensors. In this clustering approach, each cluster update is equivalent to independently updating each of the sensors that comprise it, similar to their update in the fully distributed algorithm presented above (see step \eqref{eq:update}). However, in sharp contrast to the fully distributed update, the independent sensor updates in each colored cluster can be done in parallel rather than serially. Obtaining a colored clustering can be done efficiently via a distributed graph-coloring algorithm \cite{BE2013-B}. To summarize, colored clustering results in a fully distributed and partially parallel algorithm. We demonstrate the benefits of the colored clustering approach in Section \ref{Sec:Numerics}, and therefore it deserves a deeper study in a future work.

\section{Analysis of the Unifying AM Algorithm} \label{Sec:Analysis}

\subsection{Basic Properties}
	Recalling Problem \eqref{ReformulationMatrix}
	\begin{equation*}
		\underset{\begin{smallmatrix}\bx \in \real^{nN} \\ \bu \in \BBB^{M}\end{smallmatrix}}{\min} \left\{ \bx^{T}\bpp\bx - 2\left(\bww\bu + \baa^{T}\bbb\ba\right)^{T}\bx + 2\bs^{T}\bu \right\},
	\end{equation*}
	throughout the rest of the paper we will pose the following mild assumption on the networks that we can handle.
	\begin{assumption} \label{A:AssumptionA}
		\begin{itemize}
			\item[(i)] The graph obtained by the network is connected.
			\item[(ii)] There is at least one anchor sensor, that is, $m \geq 1$.
		\end{itemize}		
	\end{assumption}
	It should be noted that if Item (i) does not hold, namely, the network has more than one connected component, then it can be divided into disjoint connected sub-networks that can be treated separately, as long as each such sub-network contains an anchor node.

	Under this assumption on the network we can obtain the following basic properties of our problem (due to the technical nature of the proof, we postpone it to the Appendix below).
	\begin{prop} \label{P:BasicPro}
		\begin{itemize}
			\item[(i)] The matrix $\bpp$ is positive definite.
			\item[(ii)] The function $\bx \rightarrow F\left(\bx , \bu\right)$ is strongly convex for any fixed $\bu$.
			\item[(iii)] The function $F\left(\cdot , \cdot\right)$ is coercive, that is, $F\left(\bx , \bu\right) \rightarrow \infty$ as $\norm{\left(\bx , \bu\right)} \rightarrow \infty$.
			\item[(iv)] Problem \eqref{ReformulationCompact} attains its optimal solution.
		\end{itemize}
	\end{prop}
	
\subsection{Global Convergence}
	Our main goal in this part is to show that the Unifying AM Algorithm presented above globally converges to a critical point of the original Problem \eqref{Problem}. As discussed in Section \ref{Sec:Formulation}, we will first show that our algorithm finds critical points $\left(\bx^{\ast} , \bu^{\ast}\right)$ of Problem \eqref{ReformulationCompact} and then the desired conclusion will follow from Proposition \ref{P:CriticalCoincide}. To this end we will rely on a recent proof methodology of first-order methods in the non-convex setting, which was developed first in \cite{AB2009,ABS2013} and was extended and simplified in \cite{BST2014} (see \cite{BSTV2018} for a recent simple and concise summary of the methodology). The main mathematical tool that stands at the heart of this proof methodology, is the Kurdyka-{\L}ojasiewicz (KL) property \cite{L1963,K1998} (see \cite{BDL2006} for an extension to non-smooth functions). 
	
	The general result in \cite{BST2014} states that if an algorithm, which is designed to solve a specific optimization problem, generates a gradient-like descent sequence (in terms of \cite[Definition 6.1, p. 2147]{BSTV2018}), then it globally converges to a critical point of the problem. A recent modification, given in \cite{STV2018}, shows that if the algorithm generates a gradient-like descent sequence only with respect to a subset of the variables (see \cite[Definition 4.2, p. 661]{STV2018}), then global convergence can be shown with respect to that subset. In our case, we are interested in proving global convergence only with respect to the original variable $\bx$. Therefore, we focus on showing that the Unifying AM Algorithm for solving Problem \eqref{ReformulationCompact} generates a sequence $\Seq{\bx}{k}$ which is indeed a gradient-like descent sequence for minimizing the objective function $F$. Our main result is as follows.
	\begin{thm} \label{T:AbstrGlob}
		Let $\left\{ \left(\bx^{k} , \bu^{k}\right) \right\}_{k \in \nn}$ be a sequence generated by the Unifying AM Algorithm. Then, the sequence $\Seq{\bx}{k}$ has a finite length, i.e., $\sum_{k = 1}^{\infty} \norm{\bx^{k + 1} - \bx^{k}} < \infty$, and it converges to some $\bx^{\ast}$, which is a critical point of Problem \eqref{Problem}.
	\end{thm}
	In order to prove this result, we first present three auxiliary lemmas. For the sake of proving the lemmas, we need to recall that the vector $\bx$ is split into $q$ blocks according to the clusters $\CCC_{1} , \CCC_{2} , \ldots , \CCC_{q}$, that is, $\bx = \left({\bar \bx}_{1} , {\bar \bx}_{2} , \ldots , {\bar \bx}_{q}\right)$. During the updating process of the Unifying AM Algorithm, we will need the following notations. Fix an iteration $k \in \nn$ and a block index $1 \leq i \leq q$, 
	\begin{equation*}
		\bx^{k , i} = \left({\bar \bx}_{1}^{k + 1} , \ldots , {\bar \bx}_{i - 1}^{k + 1} , {\bar \bx}_{i}^{k + 1} , {\bar \bx}_{i + 1}^{k} , \ldots , {\bar \bx}_{q}^{k}\right).
	\end{equation*}
	Therefore, we obviously have that $\bx^{k , 0} = \bx^{k}$ and $\bx^{k , q} = \bx^{k + 1}$. We will also need to specify the objective function $G$ (defined in \eqref{ObjG}) with respect to each block ${\bar \bx}_{i}$, $1 \leq i \leq q$, separately. We define
	\begin{equation*}
		F_{i}^{k}\left(\bxi\right) := F\left({\bar \bx}_{1}^{k + 1} , \ldots , {\bar \bx}_{i - 1}^{k + 1} , \bxi , {\bar \bx}_{i + 1}^{k} , \ldots , {\bar \bx}_{q}^{k} , \bu^{k}\right).
	\end{equation*}
	It is easy to see that $F_{i}^{k}\left(\bxi\right)$ is a quadratic function in $\bxi$, which is strongly convex (similar proof as in Proposition \ref{P:BasicPro}(ii)). Lastly, for convenience of the proofs, and to correspond with the scheme in \cite{STV2018}, we denote from now on $\bw^{k + 1} = \bu^{k}$, for all $k \in \nn$. 
	
	We begin with the first lemma that show a sufficient decrease property of the sequence of function values. 
	\begin{lemma} \label{L:Sufficient}
		Let $\left\{ \left(\bx^{k} , \bu^{k}\right) \right\}_{k \in \nn}$ be a sequence generated by the Unifying AM Algorithm. Then, there exists $\rho_{1} > 0$ such that for all $k \in \nn$, we have
		\begin{equation*}
			\rho_{1}\norm{\bx^{k + 1} - \bx^{k}}^{2} \leq F\left(\bx^{k} , \bw^{k}\right) - F\left(\bx^{k + 1} , \bw^{k + 1}\right).
		\end{equation*}
	\end{lemma}	
	\begin{proof}
		Fix $k \in \nn$. Since each $F_{i}^{k}\left(\bxi\right)$, $1 \leq i \leq q$, is $\sigma_{i}$-strongly convex (for some $\sigma_{i} > 0$), we obtain from its definition that
		\begin{align*}
			F_{i}^{k}\left({\bar \bx}_{i}^{k}\right) & \geq F_{i}^{k}\left({\bar \bx}_{i}^{k + 1}\right) + \act{\nabla F_{i}^{k}\left({\bar \bx}_{i}^{k + 1}\right) , {\bar \bx}_{i}^{k} - {\bar \bx}_{i}^{k + 1}} + \frac{\sigma_{i}}{2}\norm{{\bar \bx}_{i}^{k} - {\bar \bx}_{i}^{k + 1}}^{2} = F_{i}^{k}\left({\bar \bx}_{i}^{k + 1}\right) + \frac{\sigma_{i}}{2}\norm{{\bar \bx}_{i}^{k} - {\bar \bx}_{i}^{k + 1}}^{2},
		\end{align*}
		where the last equality follows from the fact that $\nabla F_{i}^{k}\left({\bar \bx}_{i}^{k + 1}\right) = \bo$, which follows from the optimality of ${\bar \bx}_{i}^{k + 1}$ with respect to $F_{i}^{k}$ according to the update rule of the Unifying AM Algorithm. In addition, using our compact notations, we obviously have that $F_{i}^{k}\left({\bar \bx}_{i}^{k}\right) = F\left(\bx^{k , i - 1} , \bu^{k}\right)$ and $F_{i}^{k}\left({\bar \bx}_{i}^{k + 1}\right) = F\left(\bx^{k , i} , \bu^{k}\right)$.
		Therefore, for all $1 \leq i \leq q$ we have
		\begin{equation*}
			F\left(\bx^{k , i - 1} , \bu^{k}\right) \geq F\left(\bx^{k , i} , \bu^{k}\right) + \frac{\sigma_{i}}{2}\norm{{\bar \bx}_{i}^{k} - {\bar \bx}_{i}^{k + 1}}^{2}.
		\end{equation*}
		Summing this inequality for all $1 \leq i \leq q$, yields
		\begin{align*}
			F\left(\bx^{k} , \bu^{k}\right) & = F\left(\bx^{k , 0} , \bu^{k}\right) \geq F\left(\bx^{k , q} , \bu^{k}\right) + \sum_{i = 1}^{q} \frac{\sigma_{i}}{2}\norm{{\bar \bx}_{i}^{k} - {\bar \bx}_{i}^{k + 1}}^{2} \geq F\left(\bx^{k + 1} , \bu^{k}\right) + \frac{\sigma}{2}\norm{\bx^{k} - \bx^{k + 1}}^{2},
		\end{align*}
		where $\sigma = \min \left\{ \sigma_{1} , \sigma_{2} , \ldots , \sigma_{q} \right\}$. In addition, from the updating rule of the $\bu$-block we have that
		\begin{equation*}
			F\left(\bx^{k} , \bu^{k - 1}\right) \geq F\left(\bx^{k} , \bu^{k}\right).
		\end{equation*}
		Combining the last two inequalities and using the fact that $\bw^{k + 1} = \bu^{k}$, $k \in \nn$, yields the desired result.
	\end{proof}
	An immediate consequence of the first lemma is that the Unifying AM Algorithm generates a bounded sequence.
	\begin{lemma} \label{L:Bounded}
		Let $\left\{ \left(\bx^{k} , \bu^{k}\right) \right\}_{k \in \nn}$ be a sequence generated by the Unifying AM Algorithm. Then, the sequence is bounded.
	\end{lemma}		
	\begin{proof}
		From Lemma \ref{L:Sufficient} that the sequence $\left\{ F\left(\bx^{k} , \bw^{k}\right) \right\}_{k \geq 1}$ decreases and therefore the sequence $\left\{ \left(\bx^{k} , \bw^{k}\right) \right\}_{k \geq 1}$ belongs to the level set of the function $F$ at the level $F\left(\bx^{1} , \bw^{1}\right)$. Using Proposition \ref{P:BasicPro}(iii), we know that $F$ is coercive and thus has bounded level sets \cite[Proposition 11.12, p. 158]{BC2017-B}, which completes the proof.
	\end{proof}
	The last lemma shows that at each iteration of the Unifying AM Algorithm, there exists a subgradient of the objective function $F$ that is bounded by the norm of the difference between two corresponding iterates of the $\bx$ block.
	\begin{lemma} \label{L:Bound}
		Let $\left\{ \left(\bx^{k} , \bu^{k}\right) \right\}_{k \in \nn}$ be a sequence generated by the Unifying AM Algorithm. Then, there exist a scalar $\rho_{2} > 0$ and a vector $\by^{k + 1} \in \partial F\left(\bx^{k + 1} , \bw^{k + 1}\right)$ for all $k \in \nn$, such that 
		\begin{equation*}
			\norm{\by^{k + 1}} \leq \rho_{2}\norm{\bx^{k + 1} - \bx^{k}}.
		\end{equation*}
	\end{lemma}	
	\begin{proof}	
		Let $k \in \nn$. Since $\bw^{k + 1} = \bu^{k}$, by the definition of $F\left(\cdot , \cdot\right)$ (see \eqref{ObjF}) we get
		\begin{equation*}
			\partial F\left(\bx^{k + 1} , {\bu^{k}}\right) = \nabla G\left(\bx^{k + 1} , \bu^{k}\right) + \left(\bo_{nN} , \partial \delta_{\BBB^{M}}\left(\bu^{k}\right)\right).
		\end{equation*}
		Therefore, by using the optimality condition of the updating rule for the $\bu$-block, we have that 		\begin{equation*}
			\bo_{nM} \in \nabla_{\bu} G\left(\bx^{k} , \bu^{k}\right) + \partial \delta_{\BBB^{M}}\left(\bu^{k}\right).
		\end{equation*}
		Setting,
		\begin{equation*}
			\by^{k + 1}:= \nabla G\left(\bx^{k + 1} , \bu^{k}\right) - \left(\bo_{nN} , \nabla_{\bu} G\left(\bx^{k} , \bu^{k}\right)\right),
		\end{equation*}
		we have that $\by^{k + 1} \in \partial F\left(\bx^{k + 1} , \bu^{k}\right)$. For simplicity, we define $\by^{k + 1} := \left(\by_{\bx}^{k + 1} , \by_{\bu}^{k + 1}\right)$. In addition, from the optimality condition of the updating rule of each ${\bar \bx}_{i}$-block, $1 \leq i \leq q$, we have that $\nabla_{{\bar \bx}_{i}} G\left(\bx^{k , i} , \bu^{k}\right) = \nabla_{{\bar \bx}_{i}} F\left(\bx^{k , i} , \bu^{k}\right) = \nabla F_{i}^{k}\left({\bar \bx}_{i}^{k + 1}\right) =  \bo$. Therefore,
		\begin{align*}
			\norm{\by_{\bx}^{k + 1}} & = \norm{\nabla_{\bx} G\left(\bx^{k + 1} , \bu^{k}\right)} \leq \sum_{i = 1}^{q} \norm{\nabla_{{\bar \bx}_{i}} G\left(\bx^{k + 1} , \bu^{k}\right)} = \sum_{i = 1}^{q} \norm{\nabla_{{\bar \bx}_{i}} G\left(\bx^{k + 1} , \bu^{k}\right) - \nabla_{{\bar \bx}_{i}} G\left(\bx^{k , i} , \bu^{k}\right)}.
		\end{align*} 
		Since $G\left(\cdot , \bu^{k}\right)$ (see \eqref{ReformulationMatrix}) is a quadratic function, there exist positive constants $\alpha_{i}$, $1 \leq i \leq q$, such that
		\begin{align*}
			\norm{\nabla_{{\bar \bx}_{i}} G\left(\bx^{k + 1} , \bu^{k}\right) - \nabla_{{\bar \bx}_{i}} G\left(\bx^{k , i} , \bu^{k}\right)} \leq \alpha_{i}\norm{\bx^{k + 1}  - \bx^{k , i}} \leq \alpha_{i}\norm{\bx^{k + 1}  - \bx^{k}},
		\end{align*}
		where the last inequality follows from the definition of $x^{k , i}$. Similarly, since  $\bu \rightarrow G\left(\bx , \cdot\right)$ is linear (see \eqref{ReformulationMatrix}), there exists a positive parameter $\beta > 0$ for which
		\begin{align*}
			\norm{\by_{\bu}^{k + 1}} = \norm{\nabla_{\bu} G\left(\bx^{k + 1} , \bu^{k}\right) - \nabla_{\bu} G\left(\bx^{k} , \bu^{k}\right)} \leq \beta\norm{\bx^{k + 1}  - \bx^{k}},
		\end{align*}
		Combining these bounds and setting $\rho_{2}=\sum_{i=1}^q\alpha_i+\beta > 0$, results in the following
		\begin{equation*}
			\norm{\by^{k + 1}} \leq \norm{\by_{\bar{\bx}}^{k + 1}} + \norm{\by_{\bu}^{k + 1}} \leq \rho_{2}\norm{\bx^{k + 1}  - \bx^{k}},
		\end{equation*}
		which completes the proof.
	\end{proof}
	Now, we are ready to provide the proof of our main result. \\
	
	\begin{table*}[!t]
	\begin{threeparttable}
	\caption{Computation and communication cost per iteration} \label{tbl:CC_compare}
	\centering
	{\footnotesize
	\begin{tabular}{l|llllll|ll}				
				&\multicolumn{6}{c|}{Sensor $i$} & \multicolumn{2}{c}{Total Computational Cost}\\	
				& Computational & Storage &\multicolumn{2}{c}{In Msg.}  & \multicolumn{2}{c|}{Out Msg.} &&    \\
				Method &Operations& Size &\# &size& \# & size&\multicolumn{1}{c}{Sequential}& \multicolumn{1}{c}{Parallelized} \\
				\hline
				\hline
				AM-FC &-&-&-&-&-&-&$O(nN^3+nM)$&-\\						
				AM-FD &$O(nM_i)$&$n$&$M_i$&$n$&$1$&$n$&$O(nM)$&-\\
				AM-CC & $O(nM_i)$&$n$&$M_i$&$n$&$1$&$n$&$O(nM)$&$O\left(\sum_{j=1}^q\max_{i\in\mathcal{C}_j} nM_i\right)$\\				
				SF &$O(nM_i)$&$n$&$M_i$&$n$&$1$&$n$&$O(nM)$&$O\left(\max_{i\in\mathcal{V}} nM_i\right)$\\
				AG & $O(nM_i)$&$n$&$M_i$&$n$&$1$&$n$&$O(nM)$&$O\left(\max_{i\in\mathcal{V}} nM_i\right)$\\
				ADMM-H &$O(nM_i+n^3T_i)\tnote{a}$&$nM_i$&$M_i$&$2n$&$M_i$&$2n$&$O(nM+n^3T)\tnote{b}$&$O\left(\max_{i\in\mathcal{V}\cup\mathcal{A}}nM_i +\max_{i\in\mathcal{V}} n^3T_i\right)$\\
			\end{tabular}}
		\smallskip
		\scriptsize
		\begin{tablenotes}
			\item[a] Notation $T_i$ refers to the number of iteration required by the non-convex Newton algorithm to converge. 
			\item[b] $T=\sum_{i\in\mathcal{V}} T_i$.
		\end{tablenotes}
	\end{threeparttable}
	\end{table*}
	
	{\em Proof of Theorem \ref{T:AbstrGlob}.} The proof is based on \cite[Theorem 4.3, p. 662]{STV2018}, which involves two requirements: (i) showing that $\left\{ \left(\bx^{k} , \bw^{k}\right) \right\}_{k \in \nn}$ is a bounded gradient-like descent sequence for minimizing $F$, and (ii) proving that $F$ is a semi-algebraic function. Since the function $G$ is a quadratic polynomial function (see \eqref{ObjG}) and the ball $\BBB$ is a semi-algebraic set, it follows immediately that $F$ is semi-algebraic. In order to prove the first requirement and in view of Lemmas \ref{L:Sufficient}, \ref{L:Bounded} and \ref{L:Bound}, it is left to show that (cf. \cite[Definition 4.2, p. 661]{STV2018})
	\begin{equation*}
		\limsup_{n_{k} \rightarrow \infty} F\left(\bx^{n_{k}} , \bw^{n_{k}}\right) \leq F\left(\bx^{\ast} , \bw^{\ast}\right),
	\end{equation*}
	where $\left(\bx^{\ast} , \bw^{\ast}\right)$ is a limit point of the subsequence $\left\{ \left(\bx^{n_{k}} , \bw^{n_{k}}\right) \right\}_{k \in \nn}$. Indeed, using the fact that $F\left(\bx^{k} , \bw^{k}\right) = G\left(\bx^{k} , \bw^{k}\right)$, for all $k \in \nn$, and by the  continuity of $G$ we obtain that
	\begin{align*}
		\limsup_{n_{k} \rightarrow \infty} F\left(\bx^{n_{k}} , \bw^{n_{k}}\right) = \limsup_{n_{k} \rightarrow \infty} G\left(\bx^{n_{k}} , \bw^{n_{k}}\right) = \lim_{n_{k} \rightarrow \infty} G\left(\bx^{n_{k}} , \bw^{n_{k}}\right) = G\left(\bx^{\ast} , \bw^{\ast}\right) \leq F\left(\bx^{\ast} , \bw^{\ast}\right),
	\end{align*}
	where the last inequality follows since we always have that $G\left(\bx , \bu\right) \leq F\left(\bx , \bu\right)$ for all $\bx \in \real^{nN}$ and $\bu \in \real^{nM}$. This completes the requirements and therefore we obtain from \cite[Theorem 4.3, p. 662]{STV2018} that $\Seq{\bx}{k}$ converges to some $\bx^{\ast}$, and, for any limit point $\bw^{\ast}$ of $\Seq{\bw}{k}$, $\left(\bx^{\ast} , \bw^{\ast}\right)$ is a critical point of $F$.

	 Since $\bw^{k + 1} = \bu^{k}$, $k \in \nn$, we obtain that the same statement is true for any limit point $\bu^{\ast}$ of $\Seq{\bu}{k}$, which proves that $\left(\bx^{\ast} , \bu^{\ast}\right)$ is a critical point of $F$. The result now follows from Proposition \ref{P:CriticalCoincide}. \qed
	
\section{Numerical Experiments} \label{Sec:Numerics}
	In order to provide a complete insight on the relation of our range of algorithms with the methods already available in the literature, the following algorithms are compared:
	\begin{enumerate}		
		\item AM-U-$q$, our Unifying AM Algorithm with $q$ clusters presented above in Section \ref{Sec:Unifying}. This is a semi-distributed version, where we randomly generated $q$ geographical clusters. The geographical clusters were created similarly to well-known LEACH algorithm described in \cite{heinzelman2002application}, where $q$ sensors are randomly chosen as cluster-heads and sensors are associated to the sensor with the closest cluster-head. However, in order to use only available information, the distance between the sensors and the cluster-heads are measured by the minimal number of edges between them, and in the case of a draw the sensor is associated with the cluster-head with minimal measured distance. 
		\item AM-FC, our Fully Centralized Alternating Minimization presented above in Section \ref{Sec:Algorithms}, which is equivalent to the AM-U-$q$ algorithm with $q = 1$.
		\item AM-FD, our Fully Distributed Alternating Minimization presented above in Section \ref{Sec:Algorithms}, which is equivalent to the AM-U-$q$ algorithm with $q = N$.
		\item AM-CC, the Colored Clustered Alternating Minimization, in which clusters are built from unconnected sensors using a graph coloring algorithm taken from \cite[Section 3.1]{BE2013-B}, such that the updates of each node in the clusters can be parallelized, and the algorithm is fully distributed. 
		\item SF, namely the Nesterov’s Accelerated Projected Grdaient method of \cite{CXG2015} which solves a convex relaxation version of Problem \eqref{Problem} implemented in a distributed fashion.
		\item EML, an SDP relaxation of Problem \eqref{Problem} presented in \cite{simonetto2014distributed}.
		\item ADMM-H, hybrid ADMM algorithm suggested by \cite{PE2018} which is a distributed method with an active transition from the convex relaxation suggested by \cite{CXG2015} to the non-convex model of Problem \eqref{Problem}. The ADMM-H requires several parameters, including a regularization parameter, and parameters that control the shift from the convex to the non-convex model. The parameters used in our experiments are given in Table \ref{tbl:numerical_param}.
	\end{enumerate}
	We define an iteration of each of these methods, as the period in which all sensor locations are updated. Specifically, for the fully distributed versions, each iteration consists of each sensor receiving updates messages from all its neighbors, updating its location estimation, and sending update messages to all its neighbors.
	
	We begin by discussing the computational complexity and communication requirements of each method. Table \ref{tbl:CC_compare} summarizes the computational cost, storage requirements, and communication cost (ingoing and outgoing messages) per iteration for each sensor $i$, as well as the computational cost per iteration of the parallelized and sequential implementation of each algorithm. In the methods compared we also include Nesterov's Accelerated Gradient (AG) Method \cite{N1983} which is used for initialization of some methods (see details below). For methods which are not parallelizable, such as AM-FC and AM-FD, we only give the sequential computational cost. 
	{For the centralized method AM-FC, the total sequential computation cost is derived by noting that step \eqref{AM:Stepx} requires calculating the inverse of an $N\times N$ matrix in $\real^n$ during initialization, while in step \eqref{AM:Stepu} we calculate the Euclidean norm of $M$ vectors in $\real^n$. As for the distributed methods AM-FD, AM-CC, SF and AG, each iteration of the algorithm calculates $n$ inner products of two vectors of size $M_i$ (in addition to step \eqref{AM:Stepu}). The total parallelized computational cost of AM-CC is determined by the cluster with greatest computation cost, since the method is parallelized within each cluster and sequential among the clusters. However, since SF and AG are fully parallelizable among the sensors, their total parallelized computation cost is determined by the greatest computation cost among the sensors. We note that in each iteration of ADMM-H, each non-anchor $i$ runs a non-convex Newton's method which may not converge, and the number of iterations it uses, which is denoted in the table by $T_{i}$, is unknown. Furthermore, when using the ADMM-H method, each sensor maintains a vector with a size which is proportional to the number of its neighbors, and at each iteration it sends \emph{different} messages to each neighbor, each consisting of a vector in $\real^n$. In contrast, all other distributional methods maintain only one vector in $\real^{n}$, and at each iteration each sensor broadcasts its own estimated location (a vector in $\real^{n}$) to all its neighbors, resulting in lower energy consumption and communication time.}  

	We now investigate the empirical performance of our proposed method for several networks. We use two benchmark networks available in the Standford’s Computational Optimization Laboratory web site \cite{YE1998}. The first network consists of $K = 500$ sensors with $m= 10$ anchor sensors, and the second network consists of $K = 1000$ sensors with $m = 20$ anchors. 
	{In addition, for the latter set of parameters we generated a random network for $K=1000$ in which all sensor locations (anchor and non-anchor) are randomly generated in the two dimensional box $\left[-0.5 , 0.5\right]^{2}$. Similarly, we generated four additional random networks with $K=2000$, $K=3000$, $K=5000$ and $K=10000$, and for all four networks $2\%$ of the sensors are anchors. We also used the random $K=1000$ network to create six more networks, which differ in the amount of anchors and in their average node degree. For all networks, the measurement noise for each of the distance measurements between the sensors is a Gaussian random variable with zero mean and a standard deviation of $\sigma$. Similarly to the two benchmark networks, we took $\sigma$ to be $7\%$ of the radio communication radius $r$.}
		
	For each network, a random initial point $\bx^{0}$ from a uniform distribution Unif(-0.01,0.01)$^{nN}$ is taken, and is used for all the compared methods. For the AM based-methods we always initialized $\bu^{0} = \bo_{nM}$. Note that the ADMM-H uses a two-stage approach in the sense that it first solves a relaxation of Problem \eqref{Problem}, as explained in Section \ref{Sec:Intorduction}, to obtain a more favorable starting point for solving the non-convex formulation, by using a ``smooth transition" from one stage to the other (for exact details see \cite{PE2018}). Such a two-stage approach enables reaching the vicinity of better solutions. We also utilized such an approach when applying our AM-U-$q$, AM-FD, and AM-CC algorithms. Specifically, before starting to run the algorithm we first run a fixed number of iterations (in our experiments we did $100$)  
	of the Accelerated Gradient (AG) method \cite{N1983} on the following optimization problem:	
	\begin{equation*}
		\min_{\bx \in \real^{nN}} F\left(\bx , \bu^{0}\right),
	\end{equation*}	
	where $\bu^{0}$ is given. The AG method is applicable due to the fact that, given $\bu^{0}$, the objective function is a strongly convex quadratic function of $\bx$, which has a Lipschitz continuous gradient with a parameter that can be bounded by $L=2\left(2d_{\text{max}} + m\right)$, where $d_{\text{max}}$ is the maximal degree of any non-anchor sensor in the network with respect to other non-anchors (see Table \ref{tbl:numerical_param}). We add the suffix ``-AG100'' to the methods' names to denote this initialization. Note that running the AG method on this problem can be done in a fully distributed and fully parallelized fashion, with the same communication cost per iteration as AM-FD and AM-CC (for more details see Table~\ref{tbl:CC_compare} and Section \ref{sec:ADMM_compare}). All methods were ran for $10^{3}$ iterations, in order to have the same number of communication rounds. The parameters of the networks, as well as the specific parameters used for each method, are summarized in Table~\ref{tbl:numerical_param}. 

\begin{table}
	\caption{Network and Method Parameters}\label{tbl:numerical_param}
	\centering
	\begin{tabular}{lllll|lll|l}
		\multicolumn{5}{c|}{Network Parameters} & \multicolumn{4}{c}{Method Parameters} \\	
		\multirow{2}{*}{K}&\multirow{2}{*}{m}&\multirow{2}{*}{r}&\multirow{2}{*}{$\sigma$}&Average&\multicolumn{3}{c|}{ADMM-H} & AG \\
		&&&&\multicolumn{1}{c|}{$M_i$}&$\epsilon_{c}$&$\zeta_{c}$&$\tau_{c}$ &$d_{\text{max}}$\\
		\hline
		\multicolumn{9}{c}{Benchmark}\\
		\hline
		500&10&0.3&0.02&14.15& 0.04&  0.2 & 0.015&70 \\
		1000&20&0.1&0.007&11.01& 0.003 & 0.05& 0.002 &50 \\
		\hline
		\hline
		\multicolumn{9}{c}{Random}\\
		\hline
		\multicolumn{9}{c}{Changing $K$}\\
		\hline
		1000&20&0.061&0.00427&11.09& \multirow{5}{*}{0.003} & \multirow{5}{*}{0.05}& \multirow{5}{*}{0.002} &25 \\
		2000&40&0.043&0.00301&11.22&  & &  &23 \\
		3000&60&0.035&0.00245&11.07&  & & &22 \\
		5000&100&0.029&0.00203&12.79&  & & &27 \\
		10000&200&0.025&0.00172&18.43&  & &  &37 \\
		\hline
		\multicolumn{9}{c}{Changing $m$}\\
		\hline
		985&5&0.061&0.00427&10.90& \multirow{3}{*}{0.003} & \multirow{3}{*}{0.05}& \multirow{3}{*}{0.002} &25 \\
		990&10&0.061&0.00427&10.97& && &25 \\
		1010&30&0.061&0.00427&11.19& && &25 \\	
		\hline
		\multicolumn{9}{c}{Changing $M_i$}\\
		\hline
		1000&20&0.049&0.00340&7.21& \multirow{3}{*}{0.003} & \multirow{3}{*}{0.05}& \multirow{3}{*}{0.002} &18 \\
		1000&20&0.057&0.00398&9.71& && &24 \\
		1000&20&0.067&0.00466&13.03& && &27 \\
	\end{tabular}
\end{table}
\footnotetext{We note that in the benchmark instances, the communication radius is as reported in \cite{YE1998}. However, in these instances, not all sensor pairs with distance lower than the communication radius generate edges. In fact, for all sensors the number of neighbors was truncated, and anchor sensors have neighbors which are further away than the communication radius. The values of $\sigma$ for these instances are taken from \cite{PE2018}.}

	We now discuss the criteria for comparing the methods. For each instance we generated $R = 50$ random realizations of the measured distances, and run each method on each of these inputs.
	We then compared the methods using several criteria - the average objective value (OBV) of Problem \eqref{Problem}, the root average mean squared error (RMSE), and the estimated bias. The OBV will serve as an indication of how good are the methods in solving Problem \eqref{Problem}, whereas the RMSE, given by
	\begin{equation*}
		\text{RMSE} =\sqrt{\sum_{l=1}^R\sum_{i \in \VVV} \frac{\norm{\bx_{i}^{l} - \bx_{i}^\text{real}}^2}{R}},
	\end{equation*}
	where $\bx_{i}^\text{real}$ is the real location of sensor $i$, will provide the true average approximation error of the tested method, and $\bx^l_{i}$ is the estimated location for realization $l\in \{1,2,\ldots,R\}$. Although the tested method may be biased we are going to compare the RMSE with the root of the Cramer-Rao lower bound (CRLB) (see \cite{patwari2003relative} and \cite{patwari2005locating}). For completeness, we will also provide an estimation of the bias of each method, given by:
	\begin{equation*}
	\widehat{\text{bias}} ={\sum_{l=1}^R \frac{{\bx_{i}^{l} - \bx_{i}^\text{real}}}{R}}.
	\end{equation*}
	We also compare the methods' running times, where the running times of parallelizable steps in each iteration were computed as the maximal computation time over all parallelized components. 
	
	\begin{figure}[t]
		\includegraphics[width=0.5\textwidth]{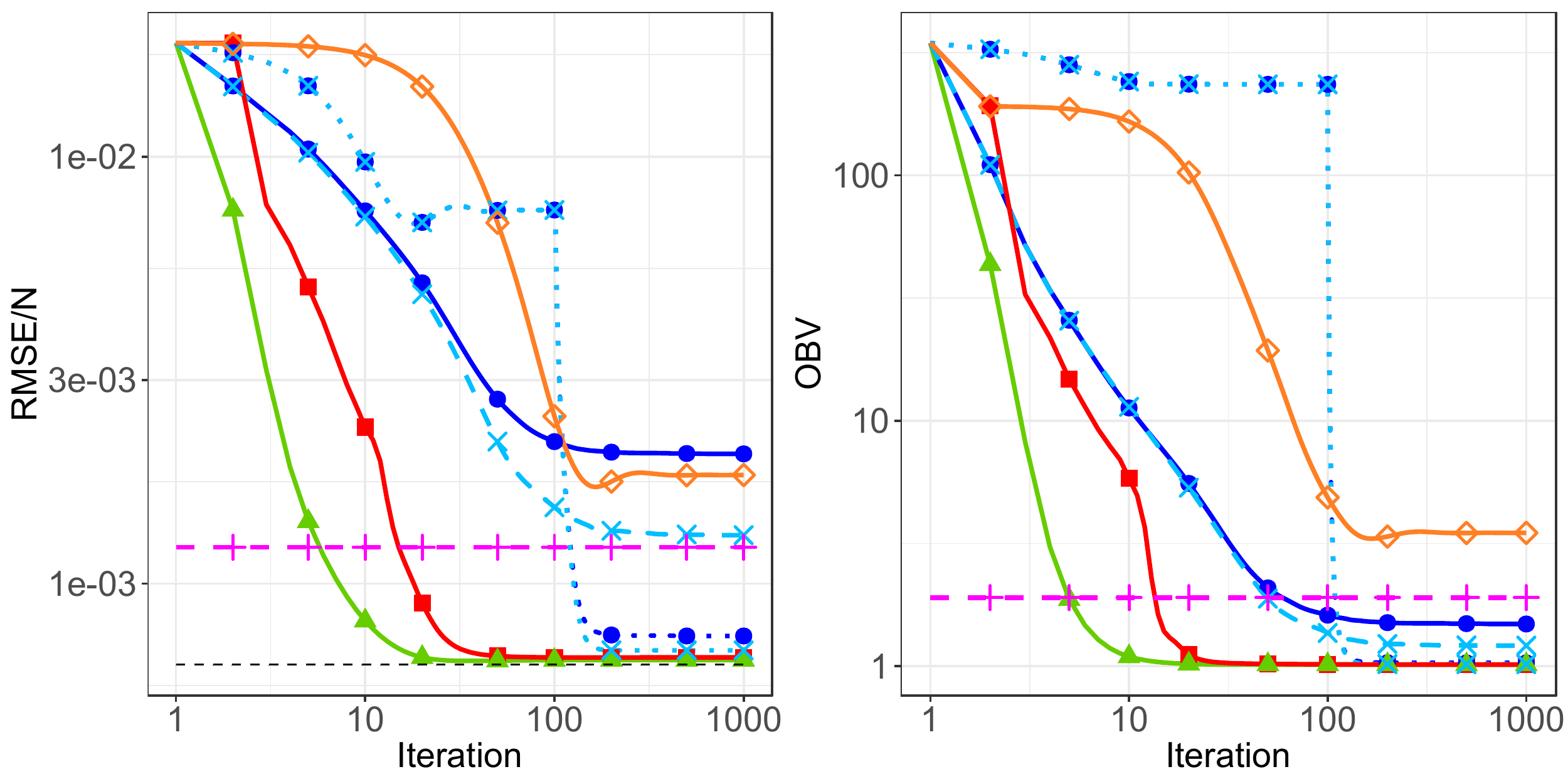}
		\includegraphics[width=0.5\textwidth]{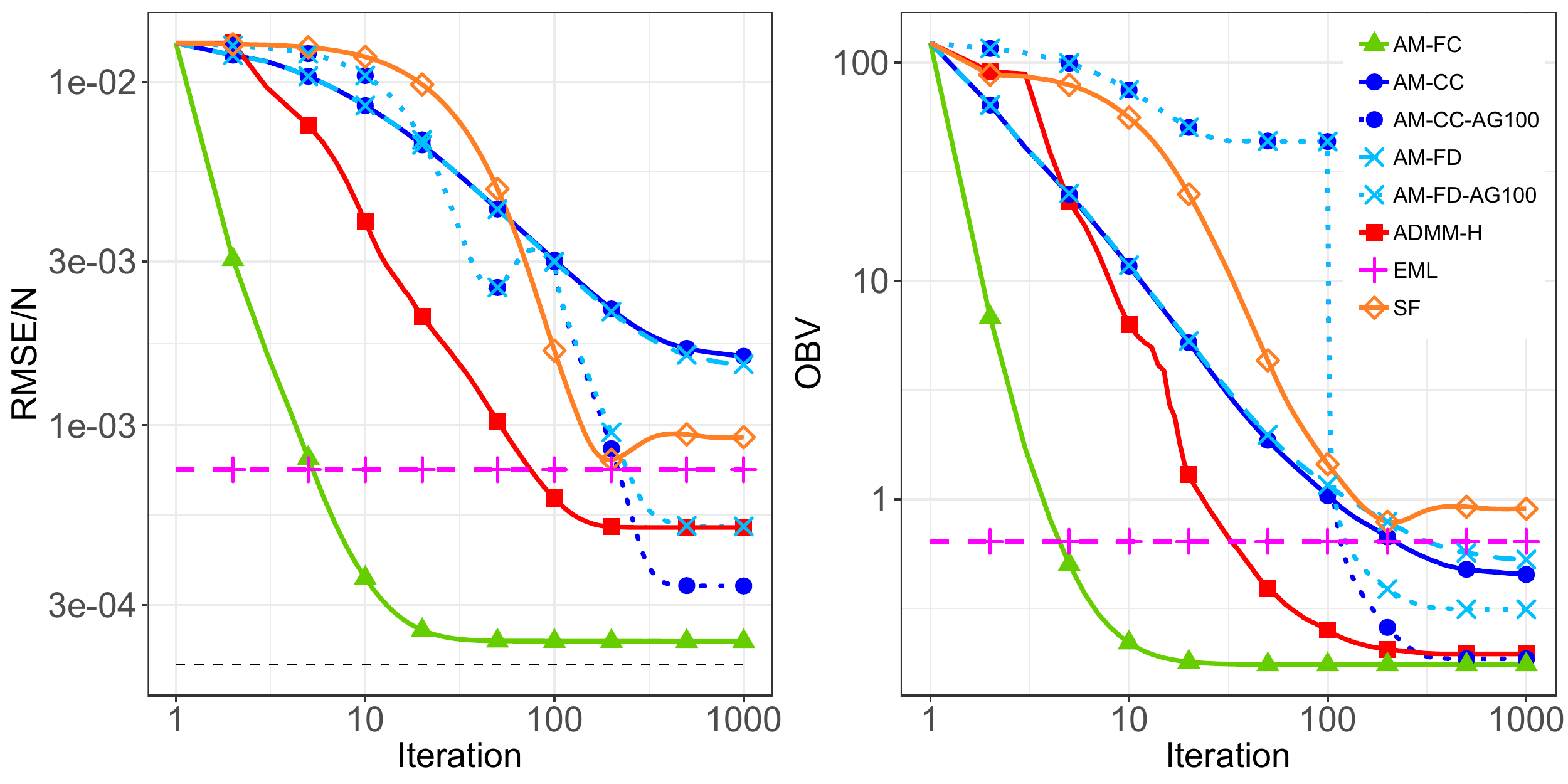}
		\caption{Comparison with convex relaxation on the benchmark networks with $K = 500$ (two left graphs) and $K=1000$ (two right graphs). The root of the CRLB divided by $N$ is given by the dashed black line. }\label{fig:relaxation_1000_benchmark}
	\end{figure}
	
	\begin{figure}[b]
		\centering
		\includegraphics[width=.7\textwidth]{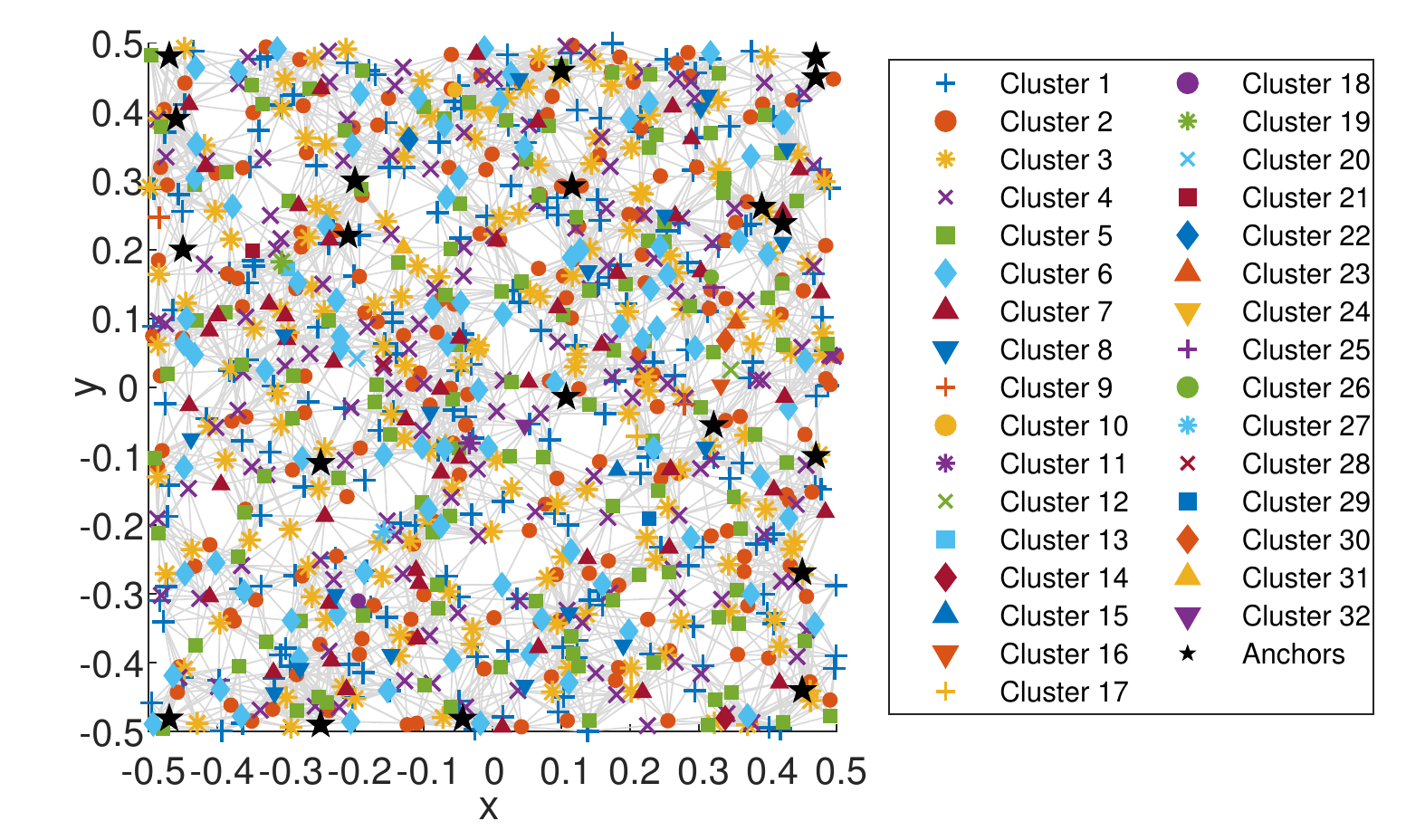}
		\caption{
			The colored clusters used by the AM-CC algorithm for the benchmark instance with $K = 1000$. There are 32 clusters, each with non-neighboring sensors, with sizes ranging from $1$ to $200$.}				
		\label{fig:colored_clustering}
	\end{figure}
	
	\begin{figure}[t]
		\centering
		\includegraphics[width=.7\textwidth]{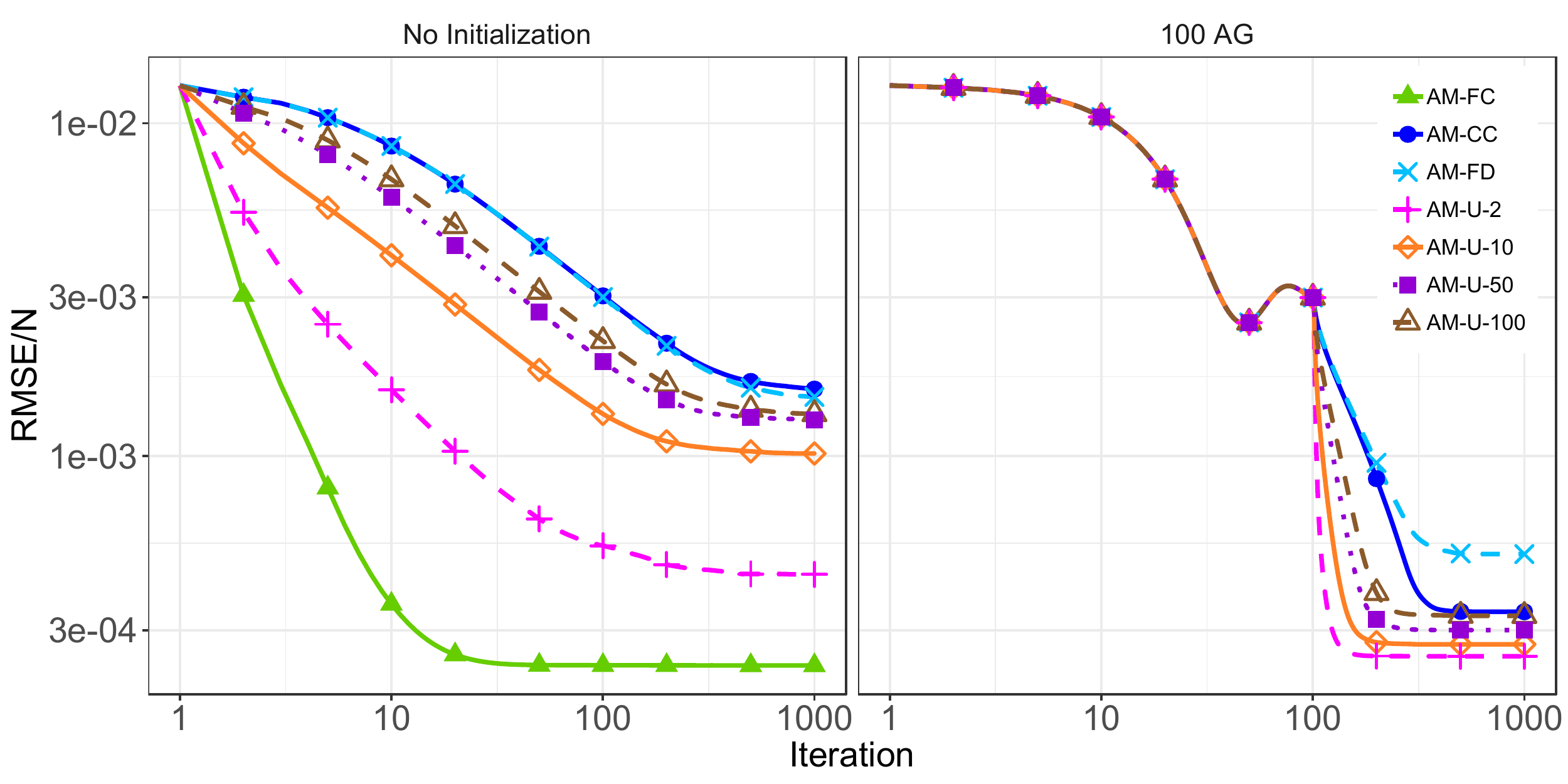}
		\caption{RMSE Comparison for AM-U-$q$ with various number of clusters on benchmark example with $K = 1000$. The left graph shows the comparison without AG initialization and the right graph shows the comparisons with AG initialization.} \label{fig:RMSE_clustering}
	\end{figure}

	All experiments were ran on an Intel(R) Xeon(R) Gold 6254 CPU @ 3.10GHz with a total of 300GB RAM and 72 threads, using MATLAB 2019a, where each method was allowed to run on only one thread and up to 16GB of RAM.

\subsection{Solving the Original Formulation vs. Relaxation}
	We begin by presenting a comparison between the centralized method AM-FC, and the distributed methods AM-FD and AM-CC, for solving Problem \eqref{Problem}, with the SF method of \cite{CXG2015} and EML of \cite{simonetto2014distributed}, which both solve different relaxations of Problem \eqref{Problem}. We note that the results for the EML method are reported as the final value obtained by the centralized algorithm, and therefore the iterations count has no significance. Indeed, while the EML has a distributed version, which is also described in \cite{simonetto2014distributed}, when attempting to run this distributed version on the benchmark networks, each iteration took approximately a second for each sensor, making this method not practical for sequential experiments. 
	The results of running $10^{3}$ iterations of these methods, starting from a common random point $\bx^0\in U(-0.01,0.01)^{2N}$  on the benchmark networks with $K=500$ and $K = 1000$, are presented in Figure~\ref{fig:relaxation_1000_benchmark}. The summary of the results after $10^3$ iterations is also available in Table~\ref{tbl:Benchmark_compare}.
	For $K=500$ and $K=1000$ the AM-CC used 127 clusters and 32 clusters, respectively, where the clusters for $K=1000$ are depicted in Figure \ref{fig:colored_clustering}. 

	We note that the AM-FC has lower RMSE than both SF and EML, illustrating that the non-relaxed problem leads to a better location estimation than the global optimal solution of the relaxed problem. In contrast, the distributed methods without initilization AM-FD and AM-CC have worse location estimation in terms of RMSE than the EML and sometimes SF method (although both the EML and SF methods have worse average function value). However, when initialized by the AG method (where the AG method is used in the first $100$ out of the $1000$ iterations) both AM-FD-100AG and AM-CC-100AG have superior RMSE and function values to EML and SF, demonstrating the benefit of using the AG method as an initialization procedure before solving the non-convex model. Moreover, the benefit of using parallelization through clustering, as done in  AM-CC and  AM-CC-100AG, is evident in the short running times of these methods, which require less than a second to run. This is in contrast to the distributed but non-parallizable AM-FD, which takes almost as much time as the centralized AM-FC, despite the problem's size. We also see that for all methods that solve the non-convex Problem~\eqref{Problem}, there is a high correlation between OBV and RMSE, whereas both EML and SF have a higher OBV than AM-CC and AM-FD while having a lower value of RMSE, suggesting that Problem \eqref{Problem} is a good predictor of the RMSE performances for low enough objective values.

	\begin{table}[b]
	\caption{Method comparison on benchmark networks}\label{tbl:Benchmark_compare}
	\centering
		\begin{tabular}{l|lllll}
			Method & RMSE & $\norm{\widehat{\text{bias}}}$ & OBV & \multicolumn{2}{c}{Avg. run time (seconds)} \\
			&Avg. (stdv)&&Avg. (stdv) &    Parallelized & Sequential \\
			\hline 
			\hline
			\multicolumn{6}{c}{Benchmark $K=500$}\\
			\hline
			AM-FC & {\bf 3.24e-01 } (1.28e-02) & 0.049 & 1.02e+00 (2.75e-02) & - & 4.28 \\ 
			EML & 5.96e-01 (2.17e-02) & 0.492 & 1.90e+00 (7.05e-02) & - & 25.28 \\ 
			AM-CC & 9.86e-01 (5.88e-02) & 0.859 & 1.49e+00 (1.05e-01) & 0.90 & 4.27\\ 
			AM-CC-AG100 & 3.69e-01 (3.54e-02) & 0.132 & 1.03e+00 (2.91e-02) & 0.87 & 3.97 \\ 
			AM-FD & 6.36e-01 (1.51e-01) & 0.330 & 1.21e+00 (1.51e-01) & - & 3.70 \\ 
			AM-FD-AG100 & 3.42e-01 (2.60e-02) & 0.071 & 1.02e+00 (2.81e-02) & - & 3.43 \\ 
			SF & 8.80e-01 (2.39e-02) & 0.798 & 3.49e+00 (1.16e-01) & 0.24 & 63.09 \\ 
			ADMM-H & {\bf 3.29e-01 }(1.36e-02) & 0.492 & {\bf 1.01e+00} (2.75e-02) & 1.40 & 452.6 \\ 
			\hline
			\multicolumn{6}{c}{Benchmark $K=1000$}\\
			\hline
			AM-FC & {\bf 2.30e-01} (1.26e-02) & 0.105 & {\bf 1.75e-01} (4.45e-03) & - & 8.99 \\ 
			EML & 7.28e-01 (1.89e-02) & 0.697 & 6.41e-01 (2.68e-02) & - & 59.43 \\ 
			AM-CC & 1.56e+00 (1.78e-01) & 1.429 & 4.52e-01 (3.96e-02) & 0.30 & 7.31 \\ 
			AM-CC-AG100 & {\bf 3.34e-01} (3.05e-02) & 0.218 & {\bf 1.86e-01} (6.96e-03) & 0.24 & 6.31 \\ 
			AM-FD & 1.48e+00 (2.18e-01) & 1.259 & 5.29e-01 (5.85e-02) & - & 8.19 \\ 
			AM-FD-AG100 & 4.98e-01 (5.95e-02) & 0.415 & 3.13e-01 (1.38e-02) & - & 7.11 \\ 
			SF & 9.05e-01 (2.12e-02) & 0.875 & 9.04e-01 (3.07e-02) & 0.23 & 107.07 \\  
			ADMM-H & 4.94e-01 (1.80e-01) & 0.308 & 1.95e-01 (1.71e-02) & 1.38 & 875.95 \\ 
			\hline
	\end{tabular}
\end{table}

\subsection{The Effect of Clustering}
	We next explore the effect of geographical clustering on the performance of our method, with and without AG initialization. We note that AG initialization is not shown for the AM-FC algorithm since no improvement was achieved by adding this initialization. 
	
	The RMSE results for $q = 2 , 10 , 50, 100 $ clusters are presented in Figure~\ref{fig:RMSE_clustering}, and in order to illustrate the output of the geographical clustering, the clusters produced for $q = 2$ and $q = 10$ are presented in Figure~\ref{fig:geo_clusters}.
	
	As one might expect, when the number of clusters decreases the algorithm becomes more centralized and produces lower RMSE values. Moreover, initializing the methods via the AG method improved the results significantly, reducing the RMSE of AM-U-2 and AM-U-10 close to that of the AM-FC. Thus, we can conclude, that if we can not run the fully centralized version AM-FC due to computational or storage restriction, there is still a benefit of using the AM-U method with some level of clustering, since together with an AG initialization, this could lead to a lower RMSE value. We note that, as shown in the previous section, we observed a high correspondence between the OBV and RMSE values for different clustering techniques, where a lower OBV implied a lower RMSE value.
		
	\begin{figure}[t]
		\centering
		\includegraphics[width=0.7\textwidth]{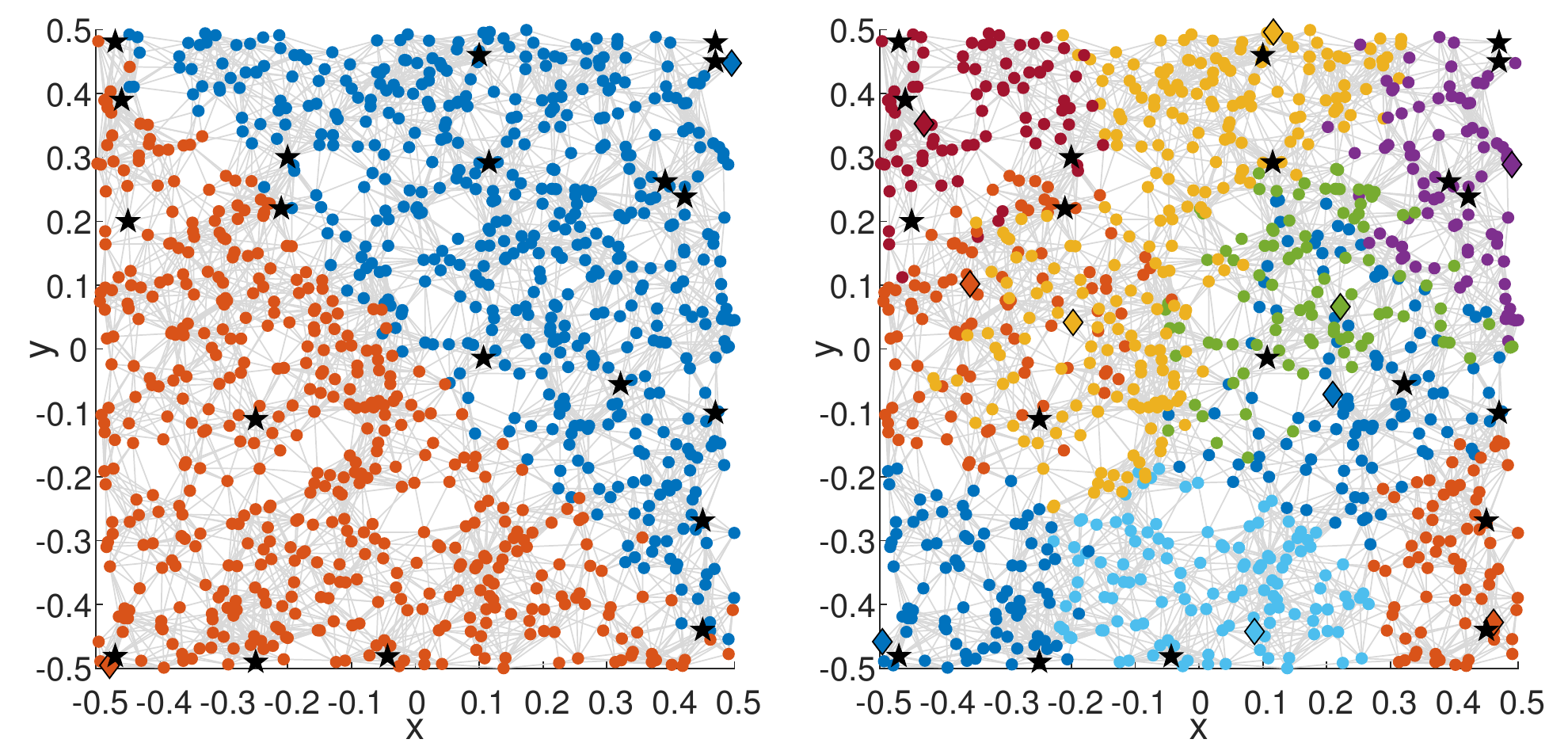}
		\caption{Results of geographical clustering for the benchmark example with $K = 1000$, with $q = 2$ clusters (right) and $q = 10$ clusters (left). The sensors which belong to each cluster are represented as dots with the same color and its cluster head is denoted by a diamond with the same color, whereas anchors are denoted by stars.} \label{fig:geo_clusters}
	\end{figure}
 
\subsection{Comparison to ADMM-H}\label{sec:ADMM_compare}
	We now compare the performance of the AM-U method to this of the ADMM-H method of \cite{PE2018}. 
	As we saw in Figure~\ref{fig:relaxation_1000_benchmark} and Table \ref{tbl:Benchmark_compare}, ADMM-H has superior performance to that of AM-CC-100AG for the benchmark with $K=500$ but inferior for the benchmark with $K=1000$. Moreover, the parallelized time of ADMM-H is five times larger than that of AM-CC-100AG.
	
	We now investigate the performance of our leading centralized and distributed methods, AM-FC and AM-CC-100AG, vs. that of ADMM-H for various settings, which are specified in Table~\ref{tbl:numerical_param}. 
	
	We start with investigating the performance for random networks of sizes $K=1000,2000,3000,5000,10000$, where the anchor number is 2\% of network size, $r$ is chosen to be the minimal radius which ensures that there exists a CRLB (Fisher Information matrix is invertible) and resulting in an average $M_i$ closest to $11$, and $\sigma$ is chosen to be 7\% of $r$. Note that high sequential running times of ADMM-H did not allow  to run it on the larger networks. We present the results of this sensitivity analysis in Figure~\ref{fig:RMSE_sizes} and Table~\ref{tbl:Random_compare}.
	We observe that the AM-CC-100AG achieves lower RMSE values than ADMM-H after about 100 iterations, and continues to improve the RMSE values with more iterations while the ADMM-H already converged.  
	Moreover, AM-FC achieves the lowest RMSE followed by AM-CC-100AG, except for the case for $K=10000$. Indeed, for the case of $K=10000$, AM-CC-100AG achieves the lowest RMSE although its objective value is higher than that of AM-FC, implying that in this case the objective function is not necessarily a good predictor of the RMSE value, which might be due to a high variance in the connectivity of the sensors due to the network's size.
	
	Next we investigate the impact of the number of anchors $m$, the average degree $M_i$ of the non-anchor nodes in the network (or equivalently the radius $r$), and the initialization of $\bx^0$ on the results. For all these experiments we take as a base-line the network with $K=1000$, and change one parameter at each time. The results are given in Figure~\ref{fig:RMSE_sensitivity} and Table~\ref{tbl:Random_compare}.
	
	When changing the number of anchors, for $m=5$ the AM-FC obtains the worst RMSE although it obtains the best function value. We deduce that in this case, Problem \eqref{Problem} is not a good predictor of the estimation ability of a method. However, as $m$ increases, the performance of all methods improve, where AM-FC again provides the lowest RMSE eventually equal to the CRLB value, followed by AM-CC-100AG.  
	When changing the degree of the network, we choose a minimal degree by choosing the minimal $R$ which ensures that the graph is connected and a maximal $R$ which results in an average $M_i$ of 13. All methods improve as the degree gets larger. However, as expected, the influence is more significant for the distributed methods. Regardless, AM-CC-100AG is superior to ADMM-H for all tested degree values. Finally, we see that different initialization of the starting point has very limited impact on AM-FC and AM-CC after 1000 iterations, while the ADMM-H algorithm improves as the variance of the initialization increases, making it less robust to different initialization procedures.
	
	\begin{table}[H]
		\caption{Method comparison on random networks}\label{tbl:Random_compare}
		\centering
		\resizebox{0.9\textwidth}{!}{
			\begin{tabular}{l|lllll}
				Method & RMSE & $\norm{\widehat{\text{bias}}}$ & OBV & \multicolumn{2}{c}{Avg. run time (seconds)} \\
				&Avg. (stdv)&&Avg. (stdv) &    Parallelized & Sequential \\
				\hline 
				\hline
				\multicolumn{6}{c}{Random $K=1000$, $m=20$, $Deg=11.09$, $\sqrt{CRLB}=$ 8.03e-01}\\
				\hline
				AM-FC & {\bf 8.21e-01} (5.54e-02) & 0.761 & {\bf 7.09e-02} (2.34e-03) & - & 9.28 \\ 
				AM-CC & 4.70e+00 (1.04e-01) & 4.652 & 2.26e-01 (1.42e-02) & 0.19 & 6.97 \\ 
				AM-CC-AG100 & {\bf 2.95e+00} (8.21e-02) & 2.890 & {\bf 1.25e-01} (4.74e-03) & 0.17 & 6.26 \\ 
				ADMM & 4.99e+00 (4.15e-02) & 4.964 & 1.53e-01 (5.58e-03) & 1.52 &  876.18\\ 
				\hline
				\multicolumn{6}{c}{Random $K=2000$}\\
				\hline
				AM-FC & {\bf 7.58e-01} (4.72e-02) & 0.703 & {\bf 7.53e-02} (2.67e-03) & - & 29.77 \\ 
				AM-CC & 3.31e+00 (5.88e-02) & 3.243 & 2.52e-01 (1.09e-02) & 0.26 &  19.15\\ 
				AM-CC-AG100 & {\bf 1.45e+00} (6.28e-02) & 1.379 & {\bf 1.07e-01} (3.58e-03) &  0.23 & 17.09\\ 
				ADMM & 3.42e+00 (2.32e-02) & 3.387 & 1.70e-01 (5.72e-03) & 1.67 & 1959.50 \\ 
				\hline
				\multicolumn{6}{c}{Random $K=3000$}\\
				\hline
				AM-FC & {\bf 5.13e-01} (2.81e-02) & 0.461 & {\bf 7.00e-02} (1.46e-03) & - & 59.95 \\ 
				AM-CC & 3.16e+00 (4.13e-02) & 3.120 & 2.31e-01 (1.10e-02) & 0.30  & 32.25\\ 
				AM-CC-AG100 & {\bf 1.29e+00} (4.36e-02) & 1.221 & {\bf 9.16e-02} (4.36e-03) & 0.27 &  29.09\\ 
				ADMM & 3.26e+00 (2.25e-02) &3.239  & 1.45e-01 (3.41e-03) & 1.94 & 2909.90\\ 
				\hline
				\multicolumn{6}{c}{Random $K=5000$}\\
				\hline
				AM-FC & {\bf 2.39e-01} (1.59e-02) & 0.169 & {\bf 9.23e-02} (9.26e-04) & - & 156.30 \\ 
				AM-CC & 3.28e+00 (4.58e-02) &  3.230& 3.13e-01 (1.08e-02) & 0.44 &  69.46\\ 
				AM-CC-AG100 & {\bf 1.30e+00} (4.12e-02) &1.271  & {\bf 1.13e-01} (2.31e-03) & 0.40 & 62.69 \\ 
				\hline
				\multicolumn{6}{c}{Random $K=10000$}\\
				\hline
				AM-FC & 1.02e+00 (1.27e-02) & 1.009 & {\bf 2.34e-01} (1.64e-03) & - & 819.25 \\ 
				AM-CC & 2.60e+00 (3.61e-02) & 2.556 & 5.17e-01 (1.85e-02) &  0.86 &247.93 \\ 
				AM-CC-AG100 & {\bf 6.72e-01} (4.20e-02) & 0.624 & {\bf 2.53e-01} (3.76e-03) & 0.78 &223.1\\ 
				\hline
				\multicolumn{6}{c}{Random $K=985$, $m=5$, $Deg=10.9$, $\sqrt{CRLB}=$ 8.19e-01}\\
				\hline
				AM-FC & 1.06e+01 (3.70e-01) &10.583  & {\bf1.02e-01} (3.93e-03) & - & 8.88 \\ 
				AM-CC & 9.93e+00 (1.09e-01) & 9.884 & 2.49e-01 (8.93e-03) & 0.19 &  6.81\\ 
				AM-CC-AG100 & {\bf 9.08e+00} (1.10e-01) & 9.040 & {\bf 1.78e-01} (8.12e-03) & 0.17 & 6.13 \\ 
				ADMM-H & 1.01e+01 (3.46e-02) & 10.077& 2.13e-01 (5.86e-03) & 1.61 & 939.04 \\ 
				\hline
				\multicolumn{6}{c}{Random $K=990$, $m=10$, $Deg=10.97$, $\sqrt{CRLB}=$ 8.07e-01}\\
				\hline
				AM-FC & {\bf 2.04e+00} (1.26e-01) & 1.992 & {\bf 7.27e-02} (2.52e-03) & - & 9.38 \\ 
				AM-CC & 6.68e+00 (1.21e-01) &6.633  &  2.49e-01 (1.40e-02) & 0.21 & 7.19 \\ 
				AM-CC-AG100 & {\bf 4.92e+00} (1.07e-01) & 4.857 & {\bf 1.66e-01} (6.24e-03) & 0.19 & 6.41 \\  
				ADMM & 7.33e+00 (3.81e-02) & 7.310 & 1.77e-01 (7.17e-03) &  1.70 & 964.16\\ 
				\hline
				\multicolumn{6}{c}{Random $K=1010$, $m=30$, $Deg=11.19$, $\sqrt{CRLB}=$ 7.99e-01}\\
				\hline
				AM-FC & {\bf 8.08e-01} (8.38e-02) & 0.743 & {\bf 7.60e-02} (2.84e-03) & - & 9.67 \\ 
				AM-CC & 3.49e+00 (6.62e-02) &  3.425& 2.15e-01 (2.06e-02) & 0.22 &  7.25\\ 
				AM-CC-AG100 & {\bf 1.57e+00} (8.89e-02) & 1.47 & {\bf 1.04e-01} (7.13e-03) & 0.19 & 6.49 \\ 
				ADMM & 3.42e+00 (4.02e-02) &  3.396& 1.19e-01 (7.41e-03) & 1.69 & 958.21\\ 
				\hline
				\multicolumn{6}{c}{Random $K=1000$, $m=20$, $Deg=7.2$}\\
				\hline
				AM-FC & {\bf 2.03e+00} (1.37e-01) &  1.918& {\bf 2.39e-02} (9.87e-04) & - & 5.90 \\ 
				AM-CC & 5.44e+00 (5.29e-02) &5.405  & 5.84e-02 (2.70e-03) & 0.11 &  4.78\\ 
				AM-CC-AG100 & {\bf 4.83e+00} (4.57e-02) & 4.792 & {\bf 4.25e-02} (2.46e-03) & 0.10 &4.30  \\ 
				ADMM & 6.00e+00 (2.78e-02) & 5.984 & 4.68e-02 (1.82e-03) & 1.28 & 811.11 \\ 
				\hline
				\multicolumn{6}{c}{Random $K=1000$, $m=20$, $Deg=9.71$, $\sqrt{CRLB}=8.07e-01$}\\
				\hline
				AM-FC & {\bf 1.00e+00} (7.22e-02) &  0.935& {\bf 5.02e-02} (1.75e-03) & - & 8.57 \\ 
				AM-CC &  4.86e+00 (7.74e-02) & 4.810 & 1.61e-01 (8.82e-03) & 0.21 &  6.77\\ 
				AM-CC-AG100 & {\bf 3.56e+00} (8.04e-02) & 3.513 & {\bf 9.36e-02} (4.33e-03) & 0.18 &6.07  \\ 
				ADMM & 5.31e+00 (3.20e-02) & 5.286 &  1.05e-01 (4.03e-03) & 1.74 &  961.34\\ 
				\hline
				\multicolumn{6}{c}{Random $K=1000$, $m=20$, $Deg=13.03$, $\sqrt{CRLB}= 1.81e-01$}\\
				\hline
				AM-FC & {\bf 6.99e-01} (8.09e-02) & 0.637 & {\bf 1.08e-01} (3.90e-03) & - & 10.83 \\ 
				AM-CC & 4.65e+00 (1.06e-01) & 4.601 & 3.66e-01 (1.82e-02) & 0.23 &  8.13\\ 
				AM-CC-AG100 & {\bf 2.69e+00} (1.83e-01) & 2.608 & {\bf 1.90e-01} (1.20e-02) & 0.21 & 7.34 \\ 
				ADMM-H & 4.71e+00 (5.06e-02) &4.678 & 2.33e-01 (9.35e-03) & 2.42 &  1027.10\\ 
		\end{tabular}}
	\end{table}
	
	In general, when the objective function of Problem \eqref{Problem} is a good predictor of the performance, the centralized method AM-FC obtains the lowest RMSE and converges after the fewest iterations. Moreover, the distributed method AM-CC-100AG generally has lower RMSE values than ADMM-H after $100$ iterations, although it takes longer to converge. We attribute this to the fact that ADMM-H has many parameters which need to be tuned well for each setting, and are therefore less robust to slight changes, including the choice of starting point. This is a further advantage of the AM-U method, which is parameter-free, and the AG initialization that relies only on parameters given by the network structure and can be easily approximated. Furthermore, the AM-CC method is at least five times faster than ADMM-H.

	\begin{figure}[t]
		\centering
		\includegraphics[width=0.8\textwidth]{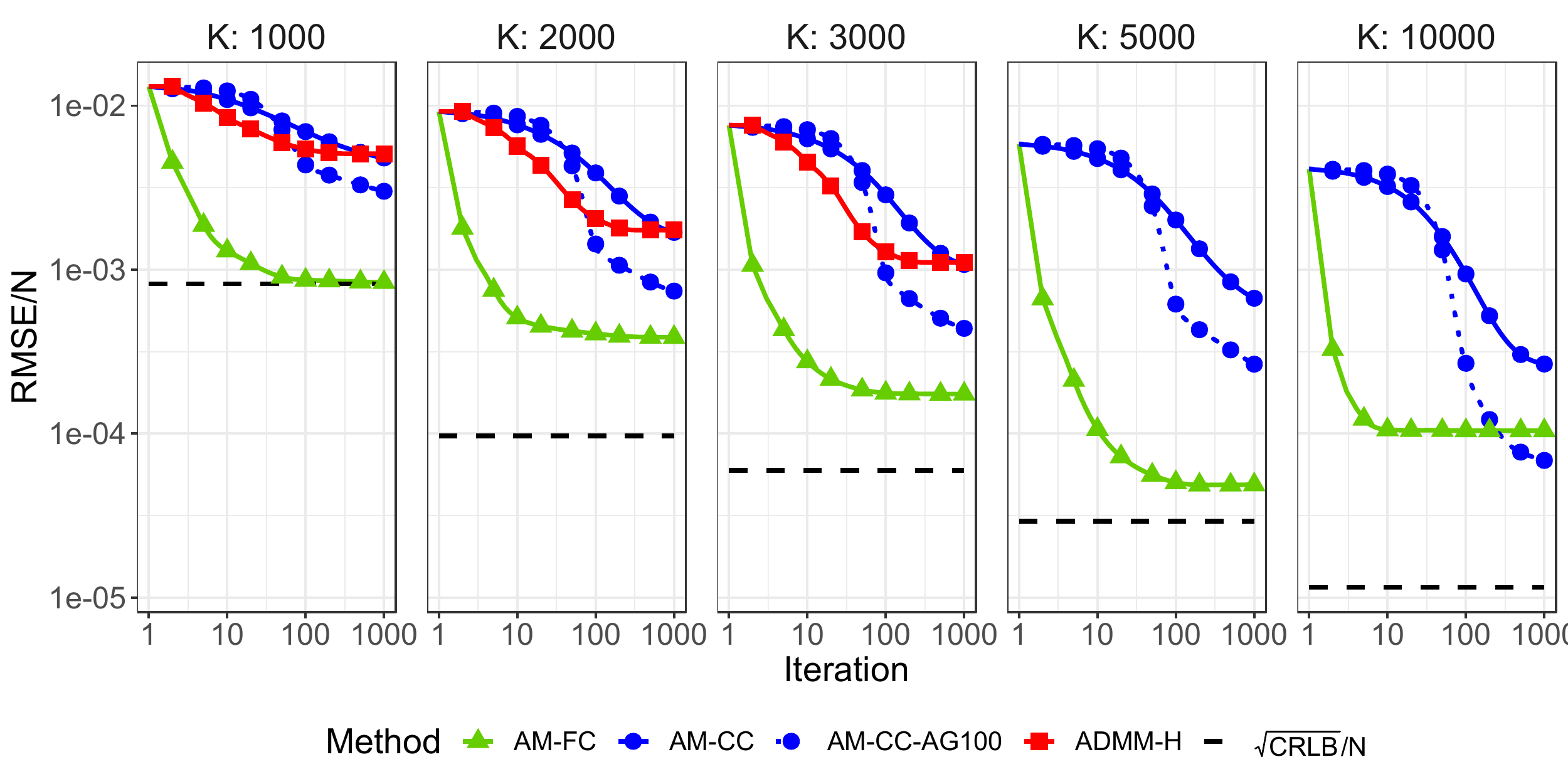}
		\caption{RMSE of AM centralized and distributed methods vs. ADMM-H for various network sizes.} \label{fig:RMSE_sizes}
	\end{figure}
	\begin{figure}[th]
		\centering
		\includegraphics[width=0.8\textwidth]{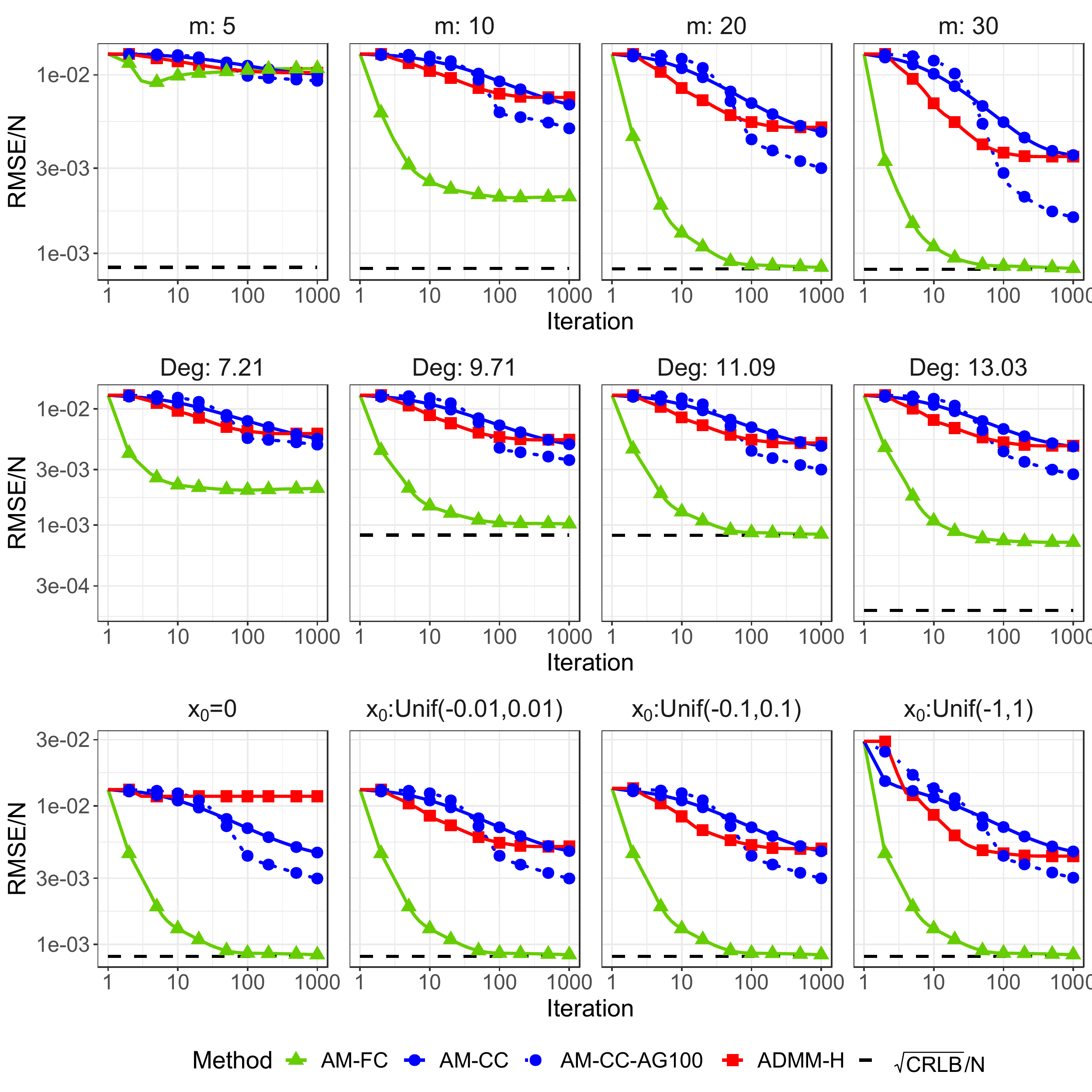}
		\caption{RMSE of AM centralized and distributed methods vs. ADMM-H on a network with $N=980$, for various number of anchors $m$ (first row), various node degree $M_i$ (second row), and various initialization of $\bx_0$ (third row).} \label{fig:RMSE_sensitivity} 
	\end{figure}

\section{Summary}
	In this paper we suggest a first-order framework based on the Alternating Minimization technique to solve the non-convex and non-smooth formulation of the WSN localization problem. This framework, provides a range of implementation, from fully distributable to fully centralized, while also allowing for partial parallelization. We prove that the algorithms generated by this general framework globally converge to critical points of the non-convex and non-smooth problem. We show in our numerical experiments that the fully centralized version is both scalable (up to 10000 sensors), and provides near optimal solution for various instances, while the distributed versions with proper initialization are still superior to existing methods. We also would like to emphasize that as far as we know, there is no other centralized method that efficiently solves the WSNL problem for such large networks.
	
\section{Acknowledgments}
We would also like to thank Nicola Piovesan for graciously providing the code for the ADMM-H algorithm from \cite{PE2018}, and Andrea Simonetto and Geert Leus for graciously providing the code for the EML algorithm from \cite{simonetto2014distributed}. Last but not least we would like to thank Amir Beck for introducing us to the WSN problem.

\begin{appendices}
	\section{Proofs of Proposition \ref{P:CriticalCoincide} and Proposition \ref{P:BasicPro}}
	
	In this section we prove the two promised results: Proposition \ref{P:CriticalCoincide} and Proposition \ref{P:BasicPro}. 
	
	\textit{Proof of Proposition \ref{P:CriticalCoincide}.} Using the conditions obtained in \eqref{OptX}, \eqref{OptU1} and \eqref{OptU2}, we obtain that $\bx_{i}^{\ast} - \bx_{j}^{\ast} \in \partial \delta_{\BBB}\left(\bu_{ij}^{\ast}\right)$ for all $\left(i , j\right) \in \EEE_{1}$ and $\bx_{i}^{\ast} - \ba_{j} \in \partial \delta_{\BBB}\left(\bu_{ij}^{\ast}\right)$ for all $\left(i , j\right) \in \EEE_{2}$. Since $\delta_{\BBB}\left(\cdot\right)$ is a proper, lower semicontinuous and convex function it follows from \cite[Theorem 4.20, p. 104]{B2017-B} that $\bu_{ij}^{\ast} \in \partial \delta_{\BBB}^{\ast}\left(\bx_{i}^{\ast} - \bx_{j}^{\ast}\right)$ for all $\left(i , j\right) \in \EEE_{1}$ and $\bu_{ij}^{\ast} \in \partial \delta_{\BBB}^{\ast}\left(\bx_{i}^{\ast} - \ba_{j}\right)$ for all $\left(i , j\right) \in \EEE_{2}$, where $\delta_{\BBB}^{\ast}$ is the Fenchel conjugate of $\delta_{\BBB}$. From \cite[Example 2.31, p. 28]{B2017-B} it also follows that $\delta_{\BBB}^{\ast}\left(\cdot\right) = \norm{\cdot}$. Hence $\bu_{ij}^{\ast} \in \partial \norm{\cdot}\left(\bx_{i}^{\ast} - \bx_{j}^{\ast}\right)$, $\left(i , j\right) \in \EEE_{1}$, and that $\bu_{ij}^{\ast} \in \partial \norm{\cdot}\left(\bx_{i}^{\ast} - \ba_{j}\right)$, $\left(i , j\right) \in \EEE_{2}$. The result now follows by using these facts in \eqref{OptX}.
\medskip
	
	\textit{Proof of Proposition \ref{P:BasicPro}.} We start by assuming, without the loss of generality, that the sub-graph built from the sensors with edges $\EEE_{1}$ is connected. While this is not necessarily the case, all the arguments given below can be applied to each connected component, and from the fact that the entire network is connected, according to Assumption \ref{A:AssumptionA}(ii), we can derive the same result.
	
	Now, in order to prove item (i), we first show that $\Ker{\tilde{\bqq}} = \Span\{\bone_{N}\}$. Indeed, let $\bv \in \Ker{\tilde{\bqq}}$. Each row of $\tilde{\bqq}$ corresponds with some $\left(i , j\right) \in \EEE_{1}$, and includes only two non zero entries, which are $1$ at the $i$-th entry and $-1$ at the $j$-th entry. Thus, we obtain that $v_{i} = v_{j}$ for all $\left(i , j\right) \in \EEE_{1}$, and since the sub-graph is connected we obtain that $\bv \in \Span\{\bone_{N}\}$. The converse inclusion trivially follows from the structure of $\tilde{\bqq}$.
		
	Now we will show that $\Ker{\tilde{\bqq}} \cap \Ker{\tilde{\baa}} = \{\bo_{N}\}$. Let $\bv \in \Ker{\tilde{\bqq}} \cap \Ker{\tilde{\baa}}$.
	Since the graph is connected (see Assumption \ref{A:AssumptionA}(i)), and since we have at least one anchor (see Assumption \ref{A:AssumptionA}(ii)), there exists a sensor $i\in\mathcal{V}$ such that $i$ is connected to an anchor $j\in\mathcal{A}$. Thus, by construction, the row associated with edge $(i,j)\in\mathcal{E}_2$ in the matrix $\tilde{\baa}$ equals to the unit vector $\be_i^T$.  Therefore, since $\bv \in \Ker{\tilde{\baa}}$ we obtain that $v_{i} = 0$, and since $\bv \in \Ker{\tilde{\bqq}} = \Span\{\bone_{N}\}$, we must have that $\bv = \bo_{N}$, which completes the proof. 
	 		
	We now show that ${\bpp}$ is positive definite. Take $\bv \in \real^{N}$, then
	\begin{align*}
		\bv^{T}\tilde{\bpp}\bv = \bv^{T}\left(\tilde{\bqq}^{T}\tilde{\bqq} + \tilde{\baa}^{T}\tilde{\baa}\right)\bv = \bv^{T}\left(\tilde{\bqq}^{T} , \tilde{\baa}^{T}\right)\begin{pmatrix} \tilde{\bqq} \\ \tilde{\baa} \end{pmatrix}\bv = \norm{\begin{pmatrix} \tilde{\bqq} \\ \tilde{\baa} \end{pmatrix}\bv}^2.
	\end{align*}
	Since $\tilde{\bpp}$ is obviously a positive semi-definite matrix and $\bv^{T}\tilde{\bpp}\bv = \bo$ if and only if $\bv \in \Ker{\tilde{\bqq}} \cap \Ker{\tilde{\baa}} = \{\bo_{N}\}$, $\tilde{\bpp}$ is positive definite. By \cite[Property IX, p. 27]{graham2018kronecker}), the eigenvalues of $\bpp$ and $\tilde{\bpp}$ are the same, and so $\bpp$ is also positive definite.
	
	Item (ii) now follows immediately from item (i) since $\bx \rightarrow G\left(\bx , \bu\right)$ is a quadratic function (see \eqref{ObjG}) for any fixed $\bu$.
	
	From \cite[Lemma 2.42, p.32]{B2014-B} it follows that $G$ is coercive and since $F \geq G$, the result of Item (iii) follows. Lastly, Item (iv) follows from \cite[Theorem 2.14, p. 20]{B2017-B}.

\end{appendices}

\printbibliography

@BOOK{B2014-B,
    AUTHOR = {A.~Beck},
     TITLE = {Introduction to nonlinear optimization},
    SERIES = {MOS-SIAM Series on Optimization},
    VOLUME = {19},
 PUBLISHER = {Society for Industrial and Applied Mathematics (SIAM), Philadelphia, PA},
      YEAR = {2014},
     PAGES = {xii+282}}

@BOOK{B2017-B,
    AUTHOR = {A.~Beck},
     TITLE = {First-Order Methods in Optimization},
    SERIES = {MOS-SIAM Series on Optimization},
    VOLUME = {25},
 PUBLISHER = {Society for Industrial and Applied Mathematics (SIAM), Philadelphia, PA},
      YEAR = {2017},
     PAGES = {xii+475}}

@BOOK{BC2017-B,
    AUTHOR = {H.~H.~Bauschke and P.~L.~Combettes},
     TITLE = {Convex Analysis and Monotone Operator Theory in {H}ilbert Spaces},
    SERIES = {CMS Books in Mathematics},
   EDITION = {Second},
 PUBLISHER = {Springer},
      YEAR = {2017},
     PAGES = {xix+619}}

@BOOK{press2007numerical,
	AUTHOR = {W.~H.~Press, and S.~A.~Teukolsky, and W.~T.~Vetterling, and B.~P.~Flannery},
     TITLE = {Numerical Recipes: The Art of Scientific Computing},
   EDITION = {Third},
 PUBLISHER = {Cambridge University Press, Cambridge},
      YEAR = {2007},
     PAGES = {xxii+1235}}

@BOOK{BE2013-B,
    AUTHOR = {L.~Barenboim and M.~Elkin},
     TITLE = {Distributed Graph Coloring: Fundamentals and Recent Developments},
    SERIES = {Synthesis Lectures on Distributed Computing Theory},
 PUBLISHER = {Morgan \& Claypool Publishers},
      YEAR = {2013},
     PAGES = {171}}

@BOOK{ammari2014art,
    AUTHOR = {H.~M.~Ammari},
     TITLE = {The Art of Wireless Sensor Networks},
    VOLUME = {1},
 PUBLISHER = {Springer-Verlag, Berlin},
      YEAR = {2014},
     PAGES = {xvii+830}}

@ARTICLE{PE2018,
    AUTHOR = {N.~Piovesan, and T.~Erseghe},
     TITLE = {Cooperative localization in {WSN}s: a hybrid convex/nonconvex solution},
   JOURNAL = {IEEE Trans. Signal Inform. Process. Netw.},
  FJOURNAL = {IEEE Transactions on Signal and Information Processing over Networks},
    VOLUME = {4},
      YEAR = {2018},
    NUMBER = {1},
     PAGES = {162--172}}

@ARTICLE{WZYS2008,
    AUTHOR = {Z.~Wang, and S.~Zheng, and Y.~Ye, and S.~Boyd},
     TITLE = {Further relaxations of the semidefinite programming approach to sensor network localization},
   JOURNAL = {SIAM J. Optim.},
  FJOURNAL = {SIAM Journal on Optimization},
    VOLUME = {19},
      YEAR = {2008},
    NUMBER = {2},
     PAGES = {655--673}}

@ARTICLE{SHCJ2010,
    AUTHOR = {Q.~Shi, and C.~He, and H.~Chen, and L.~Jiang},
     TITLE = {Distributed wireless sensor network localization via sequential greedy optimization algorithm},
   JOURNAL = {IEEE Trans. Signal Process.},
  FJOURNAL = {IEEE Transactions on Signal Processing},
    VOLUME = {58},
      YEAR = {2010},
    NUMBER = {6},
     PAGES = {3328--3340}}

@ARTICLE{BLWY2006,
    AUTHOR = {P.~Biswas, and T.-C.~Lian, and T.~C.~Wang, and Y.~Ye},
     TITLE = {Semidefinite programming based algorithms for sensor network localization},
   JOURNAL = {ACM T. Sensor Network},
   JOURNAL = {ACM Transactions on Sensor Networks},
    VOLUME = {2},
     YEAR = {2006},
    NUMBER = {2},
     PAGES = {188--220}}

@ARTICLE{LSTZ2017,
    AUTHOR = {D.~R.~Luke, and S.~Sabach, and M.~Teboulle, and K.~Zatlawey},
     TITLE = {A simple globally convergent algorithm for the nonsmooth nonconvex single source localization problem},
   JOURNAL = {J. Global Optim.},
  FJOURNAL = {Journal of Global Optimization},
    VOLUME = {69},
      YEAR = {2017},
    NUMBER = {4},
     PAGES = {889--909}}

@ARTICLE{BSTV2018,
    AUTHOR = {J.~Bolte, and S.~Sabach, and M.~Teboulle, and Y.~Vaisbourd},
     TITLE = {First order methods beyond convexity and {L}ipschitz gradient continuity with applications to quadratic inverse problems},
   JOURNAL = {SIAM J. Optim.},
  FJOURNAL = {SIAM Journal on Optimization},
    VOLUME = {28},
      YEAR = {2018},
    NUMBER = {3},
     PAGES = {2131--2151}}

@ARTICLE{BSL2008,
	AUTHOR = {A.~Beck and P.~Stoica and J.~Li},
	 TITLE = {Exact and approximate solutions of source localization problems},
   JOURNAL = {IEEE Trans. Signal Process.},
  FJOURNAL = {IEEE Transactions on Signal Processing},
    VOLUME = {56},
      YEAR = {2008},
    NUMBER = {5},
     PAGES = {1770--1778}}

@ARTICLE{CXG2015,
    AUTHOR = {C.~Soares, and J.~Xavier, and J.~Gomes},
     TITLE = {Simple and fast convex relaxation method for cooperative localization in sensor networks using range measurements},
   JOURNAL = {IEEE Trans. Signal Process.},
  FJOURNAL = {IEEE Transactions on Signal Processing},
    VOLUME = {63},
      YEAR = {2015},
    NUMBER = {17},
     PAGES = {4532--4543}}

@ARTICLE{NM2009,
    AUTHOR = {E.~Niewiadomska-Szynkiewicz, and M.~Marks},
     TITLE = {Optimization schemes for wireless sensor network localization},
   JOURNAL = {Appl. Math. Comput. Sci.}, 
  FJOURNAL = {International Journal of Applied Mathematics and Computer Science},
    VOLUME = {19},
      YEAR = {2009},
    NUMBER = {2},
     PAGES = {291--302}}

@ARTICLE{GP2016,
    AUTHOR = {S.~Goyal, and M.~S.~Patterh},
     TITLE = {Modified bat algorithm for localization of wireless sensor network},
   JOURNAL = {Wireless Pers. Commun.},
  FJOURNAL = {Wireless Personal Communications},
    VOLUME = {86},
      YEAR = {2016},
    NUMBER = {2},
     PAGES = {657--670}}

@ARTICLE{ABMS2007,
    AUTHOR = {R.~Andreani, and E.~G.~Birgin, and J.~M.~Mart\'{\i}nez, and M.~L.~Schuverdt},
     TITLE = {On augmented {L}agrangian methods with general lower-level constraints},
   JOURNAL = {SIAM J. Optim.},
  FJOURNAL = {SIAM Journal on Optimization},
    VOLUME = {18},
      YEAR = {2007},
    NUMBER = {4},
     PAGES = {1286--1309}}

@ARTICLE{N1983,
    AUTHOR = {Y.~E.~Nesterov},
     TITLE = {A method for solving the convex programming problem with convergence rate {$O(1/k\sp{2})$}},
   JOURNAL = {Dokl. Akad. Nauk SSSR},
  FJOURNAL = {Doklady Akademii Nauk SSSR},
    VOLUME = {269},
      YEAR = {1983},
    NUMBER = {3},
     PAGES = {543--547}}

@ARTICLE{BST2014,
    AUTHOR = {J.~Bolte, and S.~Sabach, and M.~Teboulle},
     TITLE = {Proximal alternating linearized minimization for nonconvex and nonsmooth problems},
   JOURNAL = {Math. Program.},
  FJOURNAL = {Mathematical Programming},
    VOLUME = {146},
      YEAR = {2014},
    NUMBER = {1-2, Ser. A},
     PAGES = {459--494}}

@ARTICLE{AB2009,
    AUTHOR = {H.~Attouch, and J.~Bolte},
     TITLE = {On the convergence of the proximal algorithm for nonsmooth functions involving analytic features},
   JOURNAL = {Math. Program.},
  FJOURNAL = {Mathematical Programming},
    VOLUME = {116},
      YEAR = {2009},
    NUMBER = {1-2, Ser. B},
     PAGES = {5--16}}

@ARTICLE{ABS2013,
    AUTHOR = {H.~Attouch, and J.~Bolte, and B.~F.~Svaiter},
     TITLE = {Convergence of descent methods for semi-algebraic and tame problems: proximal algorithms, forward-backward splitting, and regularized {G}auss-{S}eidel methods},
   JOURNAL = {Math. Program.},
  FJOURNAL = {Mathematical Programming},
    VOLUME = {137},
      YEAR = {2013},
    NUMBER = {1-2, Ser. A},
     PAGES = {91--129}}

@ARTICLE{STV2018,
    AUTHOR = {S.~Sabach, and M.~Teboulle, and S.~Voldman},
     TITLE = {A smoothing alternating minimization-based algorithm for clustering with sum-min of {E}uclidean norms},
   JOURNAL = {Pure Appl. Funct. Anal.},
  FJOURNAL = {Pure and Applied Functional Analysis},
    VOLUME = {3},
      YEAR = {2018},
    NUMBER = {4},
     PAGES = {653--679}}

@ARTICLE{BDL2006,
    AUTHOR = {J.~Bolte, and A.~Daniilidis, and A.~Lewis},
     TITLE = {The {L}ojasiewicz inequality for nonsmooth subanalytic functions with applications to
     			subgradient dynamical systems},
   JOURNAL = {SIAM J. Optim.},
  FJOURNAL = {SIAM Journal on Optimization},
    VOLUME = {17},
      YEAR = {2006},
    NUMBER = {4},
     PAGES = {1205--1223}}

@ARTICLE{simonetto2014distributed, 
	AUTHOR = {A.~Simonetto, and G.~Leus},
     TITLE = {Distributed maximum likelihood sensor network localization},
   JOURNAL = {IEEE Trans. Signal Process.},
  FJOURNAL = {IEEE Transactions on Signal Processing},
    VOLUME = {62},
      YEAR = {2014},
    NUMBER = {6},
     PAGES = {1424--1437}}

@ARTICLE{K1998,
    AUTHOR = {K.~Kurdyka},
     TITLE = {On gradients of functions definable in o-minimal structures},
   JOURNAL = {Ann. Inst. Fourier (Grenoble)},
  FJOURNAL = {Universit\'e de Grenoble. Annales de l'Institut Fourier},
    VOLUME = {48},
      YEAR = {1998},
    NUMBER = {3},
     PAGES = {769--783}}

@ARTICLE{BLTYW2006, 
	AUTHOR = {P.~Biswas, and T.-C.~Liang, and K.-C.~Toh, and Y.~Ye, and T.-C.~Wang},
	TITLE = {Semidefinite programming approaches for sensor network localization with noisy distance measurements},
	JOURNAL = {IEEE T. Autom. Sci. Eng.},
	FJOURNAL = {IEEE transactions on automation science and engineering},
	VOLUME = {3},
	YEAR = {2006},
	NUMBER = {4},
	PAGES = {360--371}}

@incollection{L1963,
    AUTHOR = {S.~{\L}ojasiewicz},
     TITLE = {Une propri\'et\'e topologique des sous-ensembles analytiques r\'eels},
 BOOKTITLE = {Les \'{E}quations aux {D}\'eriv\'ees {P}artielles ({P}aris, 1962)},
     PAGES = {87--89},
 PUBLISHER = {\'Editions du Centre National de la Recherche Scientifique, Paris},
      YEAR = {1963}}

@ARTICLE{YE1998,
  AUTHOR={Y.~Ye},
  TITLE={Computational Optimization Laboratory. Stanford University},
  URL="http://www.stanford.edu/~yyye/Col.html",
}

@inproceedings{lui2009semi,
	title={Semi-definite programming approach to sensor network node localization with anchor position uncertainty},
	author={Lui, Kenneth WK and Ma, W-K and So, Hing-Cheung and Chan, Frankie KW},
	booktitle={2009 IEEE International Conference on Acoustics, Speech and Signal Processing},
	pages={2245--2248},
	year={2009},
	organization={IEEE},
}

@book{graham2018kronecker,
	title={Kronecker products and matrix calculus with applications},
	author={Graham, Alexander},
	year={1981},
	publisher={Ellis Horwood Limited}	
}

@article{heinzelman2002application,
	title={An application-specific protocol architecture for wireless microsensor networks},
	author={Heinzelman, Wendi B and Chandrakasan, Anantha P and Balakrishnan, Hari and others},
	journal={IEEE Transactions on wireless communications},
	volume={1},
	number={4},
	pages={660--670},
	year={2002}
}

@inproceedings{agarwal2012sensor,
  title={Sensor network localization for moving sensors},
  author={Agarwal, Arvind and Daume, Hal and Phillips, Jeff M and Venkatasubramanian, Suresh},
  booktitle={2012 IEEE 12th International Conference on Data Mining Workshops},
  pages={202--209},
  year={2012},
  organization={IEEE}
}

@inproceedings{SJJ2014,
  title={Distributed, simple and stable network localization},
  author={Soares, Cl{\'a}udia and Xavier, Joao and Gomes, Joao},
  booktitle={2014 IEEE Global Conference on Signal and Information Processing (GlobalSIP)},
  pages={764--768},
  year={2014},
  organization={IEEE}
    doi={10.1109/GlobalSIP.2014.7032222},  
}

@Article{patwari2005locating,
  author    = {Patwari, Neal and Ash, Joshua N and Kyperountas, Spyros and Hero, Alfred O and Moses, Randolph L and Correal, Neiyer S},
  journal   = {IEEE Signal processing magazine},
  title     = {Locating the nodes: cooperative localization in wireless sensor networks},
  year      = {2005},
  number    = {4},
  pages     = {54--69},
  volume    = {22},
  publisher = {IEEE},
}

@Article{patwari2003relative,
  author    = {Patwari, Neal and Hero, Alfred O and Perkins, Matt and Correal, Neiyer S and O'dea, Robert J},
  journal   = {IEEE Transactions on signal processing},
  title     = {Relative location estimation in wireless sensor networks},
  year      = {2003},
  number    = {8},
  pages     = {2137--2148},
  volume    = {51},
  publisher = {IEEE},
}

\end{document}